\newif\ifarxiv
\arxivtrue
\documentclass[11pt]{article}

\usepackage{fullpage}

\usepackage{mdwlist}
\usepackage{amssymb}
\usepackage{amsmath,amscd,amsthm} % AMS Math packages
\usepackage{graphics, float}
\usepackage{fp, calc}

\usepackage[utf8]{inputenc} % allow utf-8 input
\usepackage{hyperref}       % hyperlinks
\usepackage{url}            % simple URL typesetting
\usepackage{booktabs}       % professional-quality tables
\usepackage{amsfonts}       % blackboard math symbols
\usepackage{nicefrac}       % compact symbols for 1/2, etc.
\usepackage{microtype}      % microtypography

\usepackage{algorithm}
\usepackage{algorithmic}
\usepackage{latexsym}
\usepackage{thmtools,thm-restate}
\usepackage{subcaption}
\usepackage{graphicx}
\usepackage{epstopdf}
\usepackage{hyperref}
\usepackage{xspace}
\usepackage{dsfont}
\usepackage{booktabs}
\usepackage{siunitx}
\usepackage{nicefrac}
\usepackage{pgfplots}
\usepackage{float}
\usepackage[normalem]{ulem}
\usepackage{bbm}

\usepackage{todonotes}
\usepackage{mathtools}
\DeclarePairedDelimiter{\abs}{\lvert}{\rvert}

\DeclarePairedDelimiter{\parens}{(}{)}
\DeclarePairedDelimiter{\bracks}{[}{]}

\DeclarePairedDelimiter{\ceil}{\lceil}{\rceil}

\newcommand{\EstimateMean}{\textsc{Reduced-Mean}\xspace}
\newcommand{\cD}{\mathcal{D}}
\newcommand{\ind}{\mathds{1}}
\newcommand{\Prob}[2]{\Pr_{#1}\parens*{#2}}
\newcommand{\AdaptiveSampling}{\textsc{Threshold-Sampling}\xspace}

\title{Fully Dynamic Algorithm for Constrained Submodular Optimization} 

\newtheorem{theorem}{Theorem}[section]
\newtheorem{lemma}[theorem]{Lemma}
\newtheorem{observation}[theorem]{Observation}

\newtheorem{claim}[theorem]{Claim}
\newtheorem{fact}[theorem]{Fact}

\newtheorem{definition}[theorem]{Definition}

\newtheorem{invariant}{Invariant}

\usepackage[capitalize, nameinlink]{cleveref}
\crefname{theorem}{Theorem}{Theorems}
\Crefname{lemma}{Lemma}{Lemmas}
\Crefname{fact}{Fact}{Facts}
\Crefname{claim}{Claim}{Claims}
\Crefname{observation}{Observation}{Observations}
\Crefname{invariant}{Invariant}{Invariants}

\newcommand{\OPT}{\operatorname{OPT}}
\newcommand{\pred}{\operatorname{pred}}
\newcommand{\bbRp}{{\mathbb{R}_{\ge 0}}}

\newcommand{\EE}{\operatorname{\mathbb{E}}}
\newcommand{\PP}{\operatorname{\mathbb{P}}}
\newcommand{\ee}[1]{\EE \left[ #1 \right]}
\newcommand{\prob}[1]{\PP \left[ #1 \right]}

\newcommand{\setc}[2]{\left\{#1\,\left\vert\vphantom{#1#2}\right. #2\right\}}
\newcommand{\rb}[1]{\left( #1 \right)} % round Brackets

\newcommand{\Init}{\textsc{Initialization}\xspace}
\newcommand{\Insertion}{\textsc{Insertion}\xspace}
\newcommand{\Peeling}{\textsc{Peeling}\xspace}
\newcommand{\Deletion}{\textsc{Deletion}\xspace}
\newcommand{\LevelConstruct}{\textsc{Level-Construct}\xspace}
\newcommand{\BucketConstruct}{\textsc{Bucket-Construct}\xspace}
\newcommand{\Sieve}{\textsc{SieveStreaming}\xspace}

\newcommand{\AlgSimple}{\textsc{Alg-Simple}\xspace}
\newcommand{\Random}{\textsc{Rnd}\xspace}
\newcommand{\Our}{\textsc{Alg}}
\newcommand{\OurZ}{\Our_{0.0}\xspace}
\newcommand{\OurTwo}{\Our_{0.2}\xspace}
\newcommand{\CNZ}{\textsc{CNZ}}
\newcommand{\CNZOne}{\CNZ_{0.1}\xspace}
\newcommand{\CNZTwo}{\CNZ_{0.2}\xspace}

\newcommand{\tn}{\tilde{n}}
\newcommand{\marginal}[2]{f\rb{#1 \mid #2}}

\newcommand{\ellstar}{\ell^{\star}}
\newcommand{\numtaus}{R}
\newcommand{\gopt}{\gamma}  % guess on OPT
\newcommand{\Supto}[1]{S_{\pred (#1)}}

\newcommand{\Xupto}[1]{X_{\pred (#1)}}
\newcommand{\appendixref}{\if\fullversion1 \cref{appendix:extra-experiments}\xspace \else the Appendix\xspace \fi}

\newcommand{\eqdef}{\stackrel{\text{\tiny\rm def}}{=}}
\newcommand{\SOL}{\text{Sol}}

\renewcommand{\algorithmicrequire}{\textbf{Input:}}
\renewcommand{\algorithmicensure}{\textbf{Output:}}
\renewcommand{\algorithmicforall}{\textbf{for each}}

\DeclareMathOperator{\polylog}{polylog}

\def\fullversion{1}

\title{Fully Dynamic Algorithm \\for Constrained Submodular Optimization}

\author{Silvio Lattanzi\thanks{Google Research. Email: {\tt\{silviol,ashkannorouzi,zadim\}@google.com}.} \qquad Slobodan Mitrović\thanks{MIT. Email: {\tt slobo@mit.edu}.} \qquad Ashkan Norouzi-Fard\footnotemark[1] \\ Jakub Tarnawski\thanks{Microsoft Research. Email: {\tt jakub.tarnawski@microsoft.com}.} \qquad Morteza Zadimoghaddam\footnotemark[1]}

\date{}

\begin{document}

\maketitle

%!TEX root = 00-Dynamic Submodular Maximization.tex
\begin{abstract}
The task of maximizing a monotone submodular function under a cardinality constraint is at the core of many machine learning and data mining applications, including data summarization, sparse regression and coverage problems. We study this classic problem in the fully dynamic setting, where elements can be both inserted and removed.
Our main result is a randomized algorithm that maintains an efficient data structure with a
 poly-logarithmic amortized update time  and yields a $\left(\nicefrac12 - \epsilon\right)$-approximate solution.
We complement our theoretical analysis with an empirical study of the performance of our algorithm.
\end{abstract}

\paragraph{Version v2.} This version fixes correctness issues in the previous version of this result, pointed out by the authors in~\cite{Banihashem23}; they also provide a fix. We independently provide another solution here. The main change is that in \Peeling, elements are now added one by one instead of in batches, and we explicitly check whether the marginal contribution of every added element is large enough. This implies that our claimed guarantees on the approximation ratio are now deterministic, rather than in expectation; instead, the runtime (oracle complexity) analysis has become more complex.

Since the publication of the previous version, in an exciting result, \cite{Chen22} have proved that any dynamic algorithm that maintains a better-than-$\nicefrac{1}{2}$ approximation must have an amortized query complexity that is polynomial in $n$.

%\ifarxiv\else\vspace{-1em}\fi
\section{Introduction}

Thanks to the ubiquitous nature of ``diminishing returns'' functions, submodular optimization has established itself as a central topic in machine learning, with a myriad of applications ranging from active learning~\cite{DBLP:journals/jair/GolovinK11} to sparse reconstruction~\cite{DBLP:conf/nips/Bach10, DBLP:conf/nips/DasDK12, DBLP:conf/icml/DasK11}, video analysis~\cite{DBLP:conf/nips/ZhengJCP14} and data summarization~\cite{DBLP:conf/acl/BairiIRB15}. In this field, the problem of maximizing a monotone submodular function under a cardinality constraint is perhaps the most central. Despite its generality, the problem can  be (approximately) solved using a simple and efficient greedy algorithm~\cite{nemhauser1978analysis}.

However, this classic algorithm is inefficient when applied on modern large datasets. To overcome this limitation, in recent years there has been much interest in designing efficient streaming~\cite{badanidiyuru2014streaming, chakrabarti2014submodular, DBLP:conf/soda/BuchbinderFS15a, DBLP:journals/corr/abs-1802-07098, DBLP:conf/icml/Norouzi-FardTMZ18} and distributed algorithms~\cite{mirrokni2015randomized, mirzasoleiman2015distributed, barbosa2016new, ene2019submodular} for submodular maximization. 

Although those algorithms have found numerous applications, they are not well-suited for the common applications where data is highly dynamic. In fact, real-world systems often need to handle evolving datasets, where elements are added and deleted continuously. For example, in a recent study~\cite{dey}, Dey et al.~crawled two snapshots of 1.4 million New York City Facebook users several months apart and reported that 52\% of them had changed their profile privacy settings significantly during this period. Similarly, Snapchat processes several million picture uploads and deletions daily; Twitter processes several million tweet uploads and deletions daily. As one must still be able to run basic machine learning tasks, such as sparse recovery or data summarization, in such highly dynamic settings, we need fully \emph{dynamic~algorithms}: ones able to \emph{efficiently} handle a stream containing not only insertions, but also an arbitrary number of deletions, with small processing time per  update.
%, the \emph{amortized time per update}.

The general dynamic setting is classic and a staple of algorithm design, with many applications in machine learning systems.
However, for many problems it is notoriously difficult to obtain efficient algorithms in this model.
In the case of submodular maximization,
algorithms have been proposed for the specialized settings of sliding windows~\cite{DBLP:journals/corr/ChenNZ16, DBLP:conf/www/EpastoLVZ17} and
%two-stage
robustness~\cite{DBLP:conf/nips/MitrovicBNTC17, DBLP:conf/icml/0001ZK18}.
%,
%which are special cases of the dynamic setting.
However, as we discuss below,
these solutions cannot handle the full generality of the described real-world scenarios.

\paragraph{Our contribution.}
In this paper
we design an efficient fully dynamic algorithm
for submodular maximization under a cardinality constraint.
Our algorithm:
\begin{itemize} %\ifarxiv\else \setlength\itemsep{-1pt} \fi
	\item takes as input a sequence of arbitrarily interleaved insertions and deletions,
	\item after each such update, it continuously maintains a solution whose value is
	in expectation
	at least $(\nicefrac12 - \epsilon)$ times the optimum of the underlying dataset at the current time,
	\item has amortized time per update that is poly-logarithmic in the length of the stream.
\end{itemize}
This result settles the status of submodular maximization as \emph{tractable} in the dynamic setting.
We also empirically validate the efficiency of our algorithm in several
applications.

\paragraph{Related work.}
The question of computing a concise summary of a stream of $n$ data points on the fly
was first addressed by \emph{streaming} algorithms.
This line of work focuses on using small space,
independent of (or only poly-logarithmically dependent on) $n$.
The \Sieve algorithm~\cite{badanidiyuru2014streaming}
achieves a $(\nicefrac12 - \epsilon)$-approximation in this model,
which is tight~\cite{feldman2020oneway}.
The main thresholding idea of \Sieve
%--
%carefully pick a threshold $\tau$
%and select
%each arriving item whose marginal contribution to the set of already selected items is above $\tau$
%--
has had a large influence on recent submodular works, including ours.
However, streaming algorithms do not support deletions.
In fact, the low-memory requirement is fundamentally at odds with the dynamic setting,
as any approximation algorithm for the latter must store all stream elements.\footnote{If even one element is not stored by an algorithm, an adversary could delete all other elements, bringing the approximation ratio down to $0$.}
A natural idea is to adapt streaming algorithms to deletions
by storing the stream and recomputing the solution when it loses elements.
However, this takes $\Omega(n)$ time per deletion,
and is also shown to be inefficient in our experimental evaluations.

A notable related problem
is that of maintaining a summary
that focuses only on recent data (e.g., the most recent one million data points).
This task is captured by the \emph{sliding window model}. In particular, \cite{DBLP:journals/corr/ChenNZ16, DBLP:conf/www/EpastoLVZ17} give algorithms that optimize a monotone submodular function under the additional constraint that only the last $W$ elements of the stream can be part of the solution. Unfortunately this setting, while crucial for the data freshness objective, is unrealistic for real-world dynamic systems, where it is impossible to assume that data points are deleted in such structured order. In particular, emerging privacy concerns and data protection regulations
require data processing platforms to respond rapidly to users' data removal requests. This means that the arrival and removal of data points follows an arbitrary and non-homogeneous pattern.

Another important task is that of generating a summary that is robust to a specific number $D$ of adversarial deletions. This setting is the inspiration for the \emph{two-stage deletion-robust model}. 
In the first stage, elements are inserted, and the algorithm must retain an intermediate summary of limited size.
In the second stage, an adversary removes a set of up to $D$ items. The algorithm then needs to find a final solution from the intermediate summary while excluding the removed items. Despite the generality of the deleted items being arbitrary, this framework assumes that all deletions occur after all items have been introduced to the system, which is often unrealistic and incompatible with privacy objectives.
Furthermore, in the known algorithms for this setting~\cite{DBLP:conf/nips/MitrovicBNTC17, DBLP:conf/icml/0001ZK18}, the time needed to compute a single solution depends linearly on $D$, which could be as large as the size $n$ of the entire dataset. Therefore a straightforward use of these methods in fully dynamic settings would result in $\Omega(n)$ per-update time, which is prohibitively expensive.

Finally, a closely related area is that of low-adaptivity complexity.
In particular, \cite{FMZSODA} is closely related to our work; we build upon the batch insertion idea of the Threshold Sampling algorithm introduced there.

%!TEX root = 00-Dynamic Submodular Maximization.tex
%\ifarxiv\else \vspace{-0.7em} \fi
\section{Preliminaries}
%\ifarxiv\else \vspace{-0.5em} \fi
We consider a (potentially large) collection $V$ of items, also called the \emph{ground set}. We study the problem of maximizing a \emph{non-negative monotone submodular function} $f : 2^V \to \bbRp$. Given two sets $X, Y \subseteq V$, the \emph{marginal gain} of $X$ with respect to $Y$ is defined as
\[
	\marginal{X}{Y} = f(X \cup Y) - f(Y) \,,
\]
which quantifies the increase in value when adding $X$ to $Y$. We say that $f$ is \emph{monotone} if for any element $e \in V$ and any set $Y \subseteq V$ it holds that $\marginal{e}{Y} \ge 0$. The function $f$ is \emph{submodular} if for any two sets $X$ and $Y$ such that $X \subseteq Y \subseteq V$ and any element $e \in V \setminus Y$ we have
\[
	\marginal{e}{X} \ge \marginal{e}{Y}.
\]
Throughout the paper,
we assume that
$f$ is \emph{normalized}, i.e., $f(\emptyset) = 0$.
We also assume that
$f$ is given in terms of a value oracle that computes $f(S)$ for given $S \subseteq V$.
As usual in the field, when we talk about running time, we are counting the number of oracle calls/queries, each of which we treat as a unit operation. The number of non-oracle-call operations we perform is within a polylog factor of the number of oracle calls.

\paragraph*{Submodularity under a cardinality constraint.}
The problem of maximizing a function $f$ under a \emph{cardinality constraint $k$} is defined as
selecting a set $S \subseteq V$
with $|S| \le k$
so as to maximize $f(S)$.
We will use $\OPT$ to refer to such a maximum value of $f$.

\paragraph*{Notation for dynamic streams.}
Consider a stream of insertions and deletions.
Denote by $V_i$ the set of all elements that have been inserted and not deleted up to the $i$-th operation.
Let $O_i$ be an optimum solution for $V_i$;
denote $\OPT_i = f(O_i)$.

In our dynamic algorithm we are interested in updating our data structure efficiently. We say that an algorithm has amortized update time $t$ if its total running time to process a worst-case sequence of $n$ insertions and
deletions is
in expectation
at most $nt$.

%!TEX root = 00-Dynamic Submodular Maximization.tex
\section{Overview of our approach and intuitive analysis}
\label{sec:overview}
In this section we provide an overview of the main techniques and ideas used in our algorithm. To that end we skip some details of the algorithm and present the arguments intuitively, while formal arguments are provided in \cref{sec:algorithm}. We start by noting that previous approaches either do not support deletions altogether, or support only a limited (and small) number of deletions (with linear running time per deletion) and so they do not capture many real-world scenarios.
In this work, we overcome this barrier by designing a novel fully dynamic data structure that has only poly-logarithmic amortized update time. %, which we explain in the rest of this section.
%Throughout this section, for the sake of simplicity we assume that we know the value of the optimum solution $\OPT$. 

We start with a few useful observations. For a moment, ignore the values of elements in the ground set $V$.
Consider a set $X$ of $k$ elements sampled uniformly at \emph{random} from $V$.\footnote{Hence, each element from $V$ is sampled with probability $k/n$.} The set $X$ is very robust against deletions (which, as a reminder, we asssume to be chosen \emph{independently} of the choice of $X$). Namely, in order to delete an $\epsilon$-fraction of the elements in $X$, one needs (in expectation) to delete an $\epsilon$-fraction of the elements in $V$. This property suggests the following fast algorithm, that we refer to by $\AlgSimple$, for maintaining a set of at most $k$ elements: sample $k$ elements uniformly at random from the ground set, and call that set $X$; after an $\epsilon$-fraction of the elements in $X$ is deleted, sample another $X$ from scratch. The current set $X$ represents an output after an update. \AlgSimple has expected running time $O(\nicefrac{1}{\epsilon})$ per deletion, and can also be extended to support insertions in $\polylog(n)$ time. The main issue with this approach is the lack of guarantees on the quality of the output solution after an update, i.e., the approach is oblivious to the values of the elements in $V$. For instance, the ground set might contain many \emph{useless} elements, hence selecting $k$ of them uniformly at random would not lead to a set of high utility. %Moreover, one cannot sample only good elements since the contribution of their union can be significantly less than the summation of their individual utilities, e.g., most of good elements might be a copy of the same element. 
The main idea in our paper is to partition the ground set into groups (that we call \emph{buckets}) so that applying \AlgSimple within each bucket outputs a robust set of high utility. Moreover, the union of these sampled elements provides close to optimal utility.
%Interestingly, we can show that it is possible to maintain such a partition in buckets using only poly-logarithmic time per update (insertion or deletion).

Our data structure, which we refer to by $A$, divides the elements into $T = \log n$ \emph{levels}, with each level subdivided into $R = \log k$ buckets. Informally speaking, each bucket is designed in such a way that selecting elements from it by \AlgSimple results in sets that are both robust and high-quality. Our algorithm maintains a set $S$ that represents the output solution at every point; it is constructed by applying \AlgSimple over distinct buckets. Different buckets might contribute different numbers of elements to $S$.

The structure of each level of $A$ is essentially the same. The main difference is that different levels maintain different numbers of elements, i.e., level $\ell$ maintains $O(\frac{n}{2^{\ell}} \cdot \polylog(n))$ many elements. Intuitively, and informally, levels with small $\ell$ are recomputed/changed only after many updates, while levels with large $\ell$, such as $\ell = T$, are sensitive to updates and recomputed more frequently. In particular, if we insert an extremely valuable element, then the level $\ell = T$ will guarantee that this newly added valuable element will appear in $S$. We now discuss the structure of $A$ in more detail.

%use $A_\ell$ to refer to \emph{all} the buckets in level $\ell$, and
We use $A_{i ,\ell}$ to refer to the $i$-th bucket in level $\ell$. Each level is associated with a maximum bucket-size, with level $0$ corresponding to the largest bucket-sizes. More precisely, we will maintain the invariant
\[
	|A_{i, \ell}| \leq \frac{n}{2^{\ell}} \cdot \polylog(n)
\]
for all $1 \leq i \leq R$. Organizing levels to correspond to exponentially decreasing bucket-sizes is one of the main ingredients that enables us to obtain a poly-logarithmic update time.
% We also keep a set of selected element for each level $\ell$ (which their union is the solution that achieve) denoted by $S_\ell$. 

Buckets within each level are ordered so as to contain elements of exponentially decreasing marginal values with respect to the elements chosen so far. To illustrate this partitioning, consider the first bucket of level $0$. Let $S$ be the set of elements representing our (partial) output so far; initially, $S = \emptyset$. Then, we define
\[
	A_{1, 0}= \{ e \in V \mid \tau_{1} \leq f(e\ |\ S) \leq \tau_0 \} \,,
\]
where $\tau_i \approx (1-\epsilon)^i \OPT$.\footnote{As a reminder, $\OPT$ denotes the maximum value of $f(S)$ over all $S \subseteq V$ such that $|S| \le k$.}
It is clear that the construction of $A_{1, 0}$ takes $\widetilde{O}(n)$ time. After constructing $A_{1, 0}$, our goal is to augment $S$ by some of the elements from $A_{1, 0}$ so that the marginal gain of each element added to $S$ is in expectation at least $\tau_1$. After augmenting $S$, we also refine $A_{1, 0}$. This is achieved by repeatedly performing the following steps:
\\
From $A_{1, 0}$ we randomly select a subset (of size at most $k - |S|$) of elements whose marginal gain with respect to to the solution
is at least $\tau_1$.
Then we add this set to $S$. Now, refine $A_{1, 1}$ by removing from it all elements whose marginal gain with respect to $S$ is less than $\tau_1$. If $|A_{1, 0}| \ge n/2$ and $|S| < k$, we repeat these steps. 

Let us now analyze the robustness of $S \cap A_{1, 0}$. The way we selected the elements added to $S$ enables us to perform a similar reasoning to the one we performed to analyze the robustness of $\AlgSimple$. Namely, when an element $e \in A_{1, 0}$ is added to $S$, it is always chosen \emph{uniformly at random} from $A_{1, 0}$. Also, the process of adding elements to $S$ from $A_{1, 0}$ is repeated while $|A_{1, 0}| \ge n/2$. In other words, $e$ is chosen from a large pool of elements, much larger than $k$. Hence, an adversary has to remove many elements, $\epsilon |A_{i, \ell}| \ge \epsilon n / 2$ in expectation, to remove an $\epsilon$-fraction of elements added from $A_{i, \ell}$ to $S$.
Combining this observation with the fact that the construction of $A_{1, 0}$ takes $\widetilde{O}(n)$ time is key to obtaining to the desired update time\footnote{Note that actually achieving the desired running time without any assumption requires further adjustments to the algorithm and more involved techniques that we introduce in further sections.}.
%For instance, the elements that is added to $S$ in the second iteration depends on elements added in previous iteration an so on. Arguing that this does not effect the running time requires more advanced arguments in the behavior of our algorithm. 

Note that so far we have assumed that a good solution can be constructed looking only at elements with marginal value larger than $\tau_1$. Unfortunately this is not always the case and so we need to extend our construction. To construct the remaining buckets $A_{i, \ell}$, we proceed in the same fashion as for $A_{1, 0}$ in the increasing order of $i$. The only difference is that we consider decreasing thresholds: 
\[
	A_{i,\ell} = \{ e \in V \mid \tau_{i}  \leq f(e\ |\ S) \leq \tau_{i-1} \} \,,
\]
where $S$ is always the set of elements chosen so far. Once all the buckets in level $0$ are processed, we proceed to level $1$. The main difference between different layers is that for level $\ell$ we iterate while $|A_{i, \ell}| \ge n/2^{\ell}$ and $|S| < k$. So, in every level we explore more of our ground set. Importantly, we can show that on every level we consider a ground set that decreases in size significantly.

At first, it might be surprising that from bucket to bucket of level $\ell$ we consider elements in decreasing order of their marginal gain, and then in level $\ell + 1$ we again begin by considering elements of the largest gain. Perhaps it would be more natural to first exhaust all the elements of the largest marginal gain, and only then consider those of lower gain. However, we remark that the smallest value of $\tau_i$ that we consider is at least $\Theta(\OPT/k)$. Hence, selecting for $S$ any $k$ elements whose marginal contribution is at least $\tau_i$ already leads to a good approximation.

\paragraph{Handling Deletions.} Assume that an adversary deletes an element $e$. If $e \notin S$, we remove $e$ only from the buckets it belongs to, without any extra recomputation. If $e \in S$, let $A_{i,\ell}$ be the bucket from which $e$ is added to $S$. To update $S$, we reconstruct $A$ from $A_{1,\ell}$. We now informally bound the running time needed for this reconstruction. The probability that an element from $A_{1, \ell}$ belongs to $S$ is $\frac{t}{n/2^{\ell}}$, where $t$ is the number of elements selected to $S$ from $A_{i,\ell}$. Saying it differently, an adversary has to (in expectation) remove $\frac{n/2^{\ell}}{t}$ elements from $A_{i, \ell}$ before it removes an element from $S \cap A_{i, \ell}$. Moreover the running time of a reconstruction of $A_{1, \ell}$ is $\widetilde{O}(\nicefrac{n}{2^\ell})$. Putting these two together, we get that expected running time of reconstruction per deletion is $O(t)\cdot\polylog(n)$. To reduce the update time to $\polylog{n}$, we reconstruct $S$ only if, since its last reconstruction, its value has dropped by a factor $\epsilon$. Since the elements in $A_{i, \ell}$ have similar marginal gain, an adversary would need to remove roughly $\epsilon t$ elements from $S \cap A_{i, \ell}$ to invoke a recomputation of $A_{1, \ell}$, leading to an amortized update time of $\polylog(n)$. Unfortunately, formalizing this intuition is somewhat subtle, as elements are removed from multiple buckets and each removal decreases the value of $S$.

\paragraph{Handling Insertions.} Along with $A$, we maintain buffer sets $B_1, \ldots , B_T$. Roughly speaking, our algorithm postpones processing insertions into level $\ell$ until there are $n / 2^{\ell}$ many of them; this enables us to obtain efficient amortized update time of the structure on level $\ell$. The buffer set $B_{\ell}$ is used to store these insertions until they are processed.

More precisely, when an element is inserted, the algorithm adds it to all the sets $B_{\ell}$. When, for any $\ell$, the size of $B_{\ell}$ becomes $\frac{n}{2^{\ell}}$, we add the elements of $B_{\ell}$ to $A_{1, \ell}$, reconstruct the data structure from the $\ell$-th level, and also empty $B_{\ell}$. This approach handles insertions \emph{lazily}. Notice that lazy updates should be done carefully, since if the newly inserted element has very high utility, we need to add it to the solution immediately. During the execution of the algorithm, $B_\ell$ essentially represents those elements that we have not considered in the construction of buckets in $A_{i,\ell}$ for $0 \leq \ell \leq T$. The property that the running time of constructing $A_{i, \ell}$ is $\widetilde{O}(\frac{n}{2^{\ell}})$ implies that the amortized running time per insertion is also $\polylog(n)$. Also observe that we add $B_T$ to $A_{1, T}$ after any element is inserted, which enables us to maintain a good approximate solution at all times. In particular, if an element $e$ of very large marginal gain given $S$ is inserted, e.g., $f(e\ |\ S) > \OPT / 2$, then it will be processed via $B_T$ and added to $S$. In general, if there are $2^j$ inserted elements that collectively have very large gain given $S$, then they will be processed via $B_{T - j}$ and potentially used to update $S$.

%!TEX root = 00-Dynamic Submodular Maximization.tex 
%\ifarxiv\else \vspace{-0.7em} \fi
\section{The algorithm} \label{sec:algorithm}
%\ifarxiv\else \vspace{-0.5em} \fi
We are now ready to describe our algorithm. For the sake of simplicity, we present an algorithm that is parametrized by $\gopt$:
a guess for the value $\OPT$.
Moreover we assume that we know the maximum number of elements available at any given time ($\max_{1 \leq t \leq m} |V_t|$), which is upper-bounded by $n$.
Later we show how to remove these assumptions.

%We proceed to describe the algorithm parametrized by $\gopt$.

Our algorithm maintains a data structure that uses three families of element sets: $A$ and $S$ indexed by pairs $(i,\ell)$ and $B$ indexed by $\ell$. For an integer $\numtaus$ that we will set later, the algorithm also maintains a sequence of thresholds $\tau_0 > \ldots > \tau_\numtaus$ (indexed by $i$), where we think that $\tau_0 \approx \gopt$ and $\tau_\numtaus \approx \gopt / (2k)$. 
We use $S_{j, \ell}$ to refer to the elements chosen to $S$ from bucket $j$ of level $\ell$.
Let $\Supto{i, \ell}$ be the following union of sets:
\[
	\Supto{i, \ell} \eqdef \bigcup_{1 \le j \le \numtaus, 0 \le r < \ell} S_{j, r} \cup \bigcup_{1 \le j \le i} S_{j, \ell} .
\]
In words, a set $\Supto{i, \ell}$, where ``pred'' refers to ``predecessors'', defines the subset of elements of $S$ chosen from the buckets that precede bucket $i$ of level $\ell$, including that bucket itself.
At level $\ell$ and for index $i$, we define $A_{i,\ell}$ to be the set of items with marginal value with respect to the set $\Supto{i, \ell}$ in the range $[\tau_i, \tau_{i-1}]$. While $A_{i,\ell}$ has at least $2^{T-\ell}$ items, we use a procedure called \Peeling\footnote{Algorithm \Peeling is an implementation of the ideas behind \AlgSimple described in \cref{sec:overview}.} to select a random subset of $A_{i,\ell}$ to be included into the solution set $S_{i,\ell}$. This can be done in multiple iterations; each time, a randomly chosen batch of items will be inserted into $S_{i,\ell}$.
This batch insertion logic is named $\BucketConstruct$ and summarized as \cref{alg:bucket-construct}.
The solution that our algorithm returns is $\Supto{\numtaus,T}$,
i.e., the union of all sets $S_{i,\ell}$,
and  we denote by $\SOL_t$  this set after the $t$-th operation.

In order to implement our algorithm efficiently, we need to be able to select a high-quality random subset of $A_{i,\ell}$ quickly. Our data structure enables us to do this using
the procedure \Peeling, summarized as \cref{alg:new-sampling}, which is inspired by the work~\cite{FMZSODA}.\footnote{We invoke \Peeling on the function
$f'(e) = f(e \mid \Supto{i,\ell})$, which is monotone submodular.} This procedure takes as input a set $N$ and selects a random set $S$ such that: i) the contribution of $S$ is at least $\tau |S|$,
ii) conditioned on adding $S$ to the solution, a large fraction of remaining elements in $N$ have contribution less than $\tau$,
iii) it uses only an expected logarithmic number of oracle queries per added element.

To maintain the above batch insertion logic with every insertion, the algorithm may need to recompute many of the $A$-sets, which blows up the update time. To get around this problem, we introduce buffer sets $B_\ell$ for each level $0 \leq \ell \leq T$. 
Each buffer set $B_\ell$ has a capacity of at most $2^{T-\ell}-1$ items. 
When a new item $x$ arrives, instead of recomputing all $A$-sets, we insert $x$ into all buffer sets. If some  buffer sets exceed their capacity, we pick the first one (with the smallest $\ell^*$) and reconstruct all sets in levels beginning from $\ell^*$. We call this reconstruction process $\LevelConstruct$. It is presented as \cref{alg:level-construct}. The insertion process in summarized as \cref{alg:insertion}. 
 
When deleting an element $x$, our data structure is not affected if the deleted item $x$ does not belong to any set $S_{i,\ell}$. But if it is deleted from some $S_{i,\ell}$, we need to recompute the data structure starting from $S_{i,\ell}$. To optimize the update time, we perform this update operation in a lazy manner as well. We recompute only if an $\epsilon$-fraction of items in $S_{i,\ell}$ have been deleted since the last time it was constructed. To simplify the algorithm, we reconstruct the entire level $\ell$ and also the next levels $\ell+1, ...$ in this case. The deletion logic is summarized as \cref{alg:deletion}.

We initialize all sets as empty. The sequence of thresholds $\tau$ is set up as a geometric series parametrized by a constant  $\epsilon_1 > 0$. 

	\begin{algorithm}[H]
		\caption{\Init} \label{alg:initialization}
		\begin{algorithmic}[1]
			\STATE $\numtaus \leftarrow \log_{1+\epsilon_1} (2k)$ \label{line:define-numtaus}
			\STATE $\tau_i \leftarrow \gopt (1+\epsilon_1)^{-i} \quad \forall 0 \leq i \leq \numtaus$ \label{line:define-tau}
			\STATE $T \leftarrow \log n$ \label{line:definition-of-T}
			\STATE $A_{i,\ell} \leftarrow \emptyset \quad \forall  1 \leq i \leq \numtaus \quad 0 \leq \ell \leq T$
			\STATE $S_{i,\ell} \leftarrow \emptyset \quad \forall  1 \leq i \leq \numtaus \quad 0 \leq \ell \leq T$
			\STATE $B_{\ell} \gets \emptyset \quad \forall 0 \leq \ell \leq T$
		\end{algorithmic}
	\end{algorithm}	
	\begin{algorithm}[H]
		\caption{$\Insertion(e)$} \label{alg:insertion}
		\begin{algorithmic}[1]
			\STATE $B_\ell \gets B_\ell \cup \{e\} \quad \forall 0 \leq \ell \leq T$ \label{line:insertion-add-e}
			\STATE  $V \gets V \cup \{e\}$

			\IF{there exists an index $\ell$ such that $|B_\ell| \ge 2^{T - \ell}$ \label{line:insertion-large-B-ell}}
			\STATE Let $\ellstar$ be the smallest such index \label{line:choice-of-ellstar}
			\STATE $S_{i',\ell'} \leftarrow \emptyset \quad \forall \ellstar \leq \ell' \leq T \quad \forall 1 \leq i'\leq \numtaus$ \label{line:empty-S-insertion}
			\STATE $B_{\ell'} \gets \emptyset \quad \forall \ellstar \leq \ell' \leq T$ \label{line:insert-reset-all-B-ell}
			\STATE $\LevelConstruct(\ellstar)$ \label{line:insertion-invokes-level-construct}
			\ENDIF
		\end{algorithmic}
	\end{algorithm}
	\begin{algorithm}[H]
		\caption{$\Deletion(e)$} \label{alg:deletion}
		\begin{algorithmic}[1]
			\STATE $A_{i, \ell} \gets A_{i, \ell} \setminus \{e\} \quad \forall  1 \leq i \leq \numtaus \quad 0 \leq \ell \leq T$ \label{line:remove-from-A}
			\STATE $B_\ell \gets B_\ell \setminus \{e\} \quad \forall 0 \leq \ell \leq T$ \label{line:remove-from-B}
			\STATE $V \gets V \setminus \{e\}$
			\IF {$e \in \Supto{\numtaus, T}$}
			\STATE Let $S_{i,\ell}$ be the set containing $e$
			\STATE Remove $e$ from $S_{i,\ell}$
			\IF {the size of $S_{i,\ell}$ has reduced by $\epsilon$ fraction since it was constructed \label{line:S-i-ell-reduced}}
			\STATE $S_{i',\ell'} \gets \emptyset \quad \forall \ell \leq \ell' \leq T \quad \forall 1 \leq i'\leq \numtaus$\label{line:empty-S-deletion}
			\STATE $\LevelConstruct(\ell)$ \label{line:deletion-level-construct}
			\ENDIF
			\ENDIF
		\end{algorithmic}
	\end{algorithm}
	\begin{algorithm}[H]
		\caption{$\LevelConstruct(\ell)$} \label{alg:level-construct}
		\begin{algorithmic}[1]
				\STATE $B_{\ell} \gets \emptyset$ \label{line:set-B-ell-empty}
	        \FOR{$i \gets 1 \ldots \numtaus$ \label{line:level-construct-loop}}
	            \IF{$\ell > 0$}
	            	\STATE $A_{i,\ell} \leftarrow B_{\ell - 1} \cup \bigcup_{j=1}^{R} A_{j, \ell-1}$ \label{line:previous-elements-l1}
	            \ELSE
		            \STATE $A_{i,\ell} \leftarrow V$ \label{line:previous-elements-l0}
		            
	            \ENDIF
	            \STATE $\BucketConstruct(i, \ell)$ \label{line:invoke-bucket-construct}
	        \ENDFOR
					
				\IF{$|\Supto{\numtaus, T}| \ge k$}
	        		\STATE $A_{i, \ell'} \gets \emptyset \quad \forall \ell \le \ell' \leq T\quad \forall  1 \leq i \leq \numtaus$ \label{line:reset-all-A-i-ell}
	        \ENDIF
					
			  \IF{$\ell < T$ and $|\Supto{\numtaus, T}| < k$ \label{lvlstop}}
	        		\STATE $\LevelConstruct(\ell + 1)$ \label{line:recursive-calls-level-construct}
	        \ENDIF
	    \end{algorithmic}
	\end{algorithm}
	\begin{algorithm}[H]
		\caption{$\BucketConstruct(i, \ell)$} \label{alg:bucket-construct}
		\begin{algorithmic}[1]
	        \REPEAT \label{line:repeat-bucket-construct}
	            \STATE $A_{i,\ell} = \{e \in A_{i,\ell} \ | \ \tau_i\leq f(e \ | \ \Supto{i, \ell}) \leq \tau_{i-1} \} $ \label{BCremoving}
	            \IF{$|A_{i, \ell}| \geq  2^{T-\ell}$ and $|\Supto{\numtaus, T}| < k$}
	                \STATE $S_{i,\ell} \leftarrow S_{i,\ell} \cup \Peeling (A_{i,\ell},f',k-|\Supto{\numtaus, T}|,\tau_i)$ \label{line:bucket-invoke-peeling}
	            \ENDIF
	        \UNTIL {$|A_{i, \ell}| <  2^{T-\ell}$ or $|\Supto{\numtaus, T}| \geq k$	 \label{line:condition-bucket-construct}}
	    \end{algorithmic}
	\end{algorithm}

To present \Peeling,
which is a refined variant of a building block in the algorithm 
\AdaptiveSampling of \cite{FMZSODA},
we need to define a distribution.
Let $N \subseteq V$ be some set of elements.

\begin{definition}
\label{def:indicator_distribution}
Given a set function $f : 2^N \rightarrow R^{\geq 0}$ and a subset $S \subset N$, 
let $x$ be an element chosen uniformly at random from $N$, and 
let $\mathcal{D}(N, S, f)$ denote the probability distribution over the indicator random variable
\begin{align*}
  I_S = \ind\bracks*{\marginal{x}{S} \ge \tau}.
\end{align*} 
\end{definition}

\Peeling makes use of a subroutine \EstimateMean, which is a standard unbiased
estimator for the mean of a Bernoulli distribution.
Its definition and properties are given in \cref{app:peeling}.
Intuitively, \EstimateMean attempts to answer the question: has the mean of the input distribution
(fraction of candidate elements with high enough marginal contribution to the current solution)
reduced below roughly $1-\epsilon_p/2$?

\begin{algorithm}
  \caption{\Peeling}
  \label{alg:new-sampling}
  \vspace{0.1cm}
  \textbf{Input:} Subset of items $N \subseteq V$, function $f : 2^N \rightarrow R^{\geq 0}$, constraint $k$, threshold $\tau$,
    error $\epsilon_p$
  \begin{algorithmic}[1]
    \STATE Set smaller error $\hat{\epsilon_p} \leftarrow \epsilon_p/3$
    \STATE $S \gets \emptyset$
    \WHILE{$|S| < k$}
        \IF{$\EstimateMean(\mathcal{D}(N, S, f), \hat{\epsilon_p})$} 
            \STATE \textbf{break}
        \ENDIF
        \STATE Sample element $x$ uniformly at random from $N$
        \IF{$\marginal{x}{S} \geq \tau$}
            \STATE $N \gets N \setminus \{x\}$
            \STATE $S \gets S \cup \{x\}$ 
        \ENDIF
    \ENDWHILE
    \STATE \textbf{return} $S$
  \end{algorithmic}
\end{algorithm}

%!TEX root = 00-Dynamic Submodular Maximization.tex
%\ifarxiv\else \vspace{-0.7em} \fi
\section{Analysis of the algorithm}
%\ifarxiv\else \vspace{-0.5em} \fi

We now state two technical theorems, and in \cref{sec:combining} we show how to combine them into the main result. Here $\epsilon, \epsilon_1, \epsilon_p > 0$ are parameters of our algorithm; they affect both approximation ratio and oracle complexity.
Intuitively, they should be thought of as small constants.
%(As a reminder, our approach consists of five methods \Init, \BucketConstruct, \Insertion, \Deletion and \LevelConstruct that are given as \cref{alg:initialization} through \cref{alg:level-construct}.)

\begin{restatable}{theorem}{theoremupperbound}\label{theorem:upper-bound}
Let $\SOL_i$ be the solution of our algorithm and $\OPT_i$ be the optimal solution value after $i$ updates. Moreover, assume that $\gopt$ in \cref{alg:initialization} is such that $(1+\epsilon_p)\OPT_i \geq \gopt \geq \OPT_i$. Then  for any $1 \leq i \leq n$  we have $f(\SOL_i) \geq(1-\epsilon_p -\epsilon(1+\epsilon_1)) \frac{\OPT_i}{2}$.
\end{restatable}
\begin{restatable}{theorem}{theoremoraclecomplexity}\label{theorem:oracle-complexity}
	The amortized expected number of oracle queries per update is $O\left(\frac{\numtaus^3 \log^2 n}{\epsilon_p^2 \cdot \epsilon}\right)$, where $\numtaus$ equals $\log_{1+\epsilon_1} (2k)$ (see \cref{alg:initialization}).
\end{restatable}
\cref{theorem:upper-bound,theorem:oracle-complexity} are proved in
\cref{app:correctness,app:complexity} respectively.
They hold under the assumption that the algorithm knows a good estimate $\gamma$ of $\OPT_i$ (as in the statement of \cref{theorem:upper-bound}).
To remove this assumption,
we combine these ingredients with certain standard techniques and achieve the following result. Its proof is provided in \cref{sec:combining}. 
\begin{restatable}{theorem}{mainresult}\label{thm:correctness-main}
There is a dynamic algorithm for the monotone submodular maximization under a cardinality constraint problem
that
maintains a $(1-2\epsilon_p-\epsilon(1+\epsilon_1))/2$-approximate solution after each operation,
and whose amortized expected number of oracle queries per update is $O\left(\frac{\log^4 k \log^2 n}{\epsilon_1^3 \cdot \epsilon_p^3 \cdot \epsilon}\right)$.
\end{restatable}

%\input{45-peeling}
%\input{46-puttingtogether}
% !TEX root = 00-Dynamic Submodular Maximization.tex

\section{Empirical evaluation}
\label{sec:experiments}
%!TEX root = 00-Dynamic Submodular Maximization.tex
%compat=1.13,
\pgfplotsset{
	ignore legend/.style={every axis legend/.code={\renewcommand\addlegendentry[2][]{}}},
	error bars/y explicit,
	error bars/error bar style={solid}
}

\newlength{\plotwidth}
\setlength{\plotwidth}{6.5cm}
\newlength{\plotheight}
\setlength{\plotheight}{3.5cm}
\newlength{\plotxspacing}
\setlength{\plotxspacing}{0.9cm}

\pgfplotsset{resultplot/.style={%
  scale only axis,
  yticklabel shift={-2pt},
  enlarge x limits=false,
%  xlabel shift=-5pt,
  ylabel shift=-5pt,
%  label style={font=\scriptsize},
%  title style={font=\scriptsize,yshift=-4pt},
  grid=major,
  grid style={dotted,gray,thin},
  legend cell align=left,
  legend transposed=true,
  legend columns=-1,
  legend style={font=\small,anchor=north,at={(1.2,-0.3)}},
  every axis plot/.append style={line width=1pt,mark size=2pt},
  cycle list={{blue,mark=square},{red,mark=asterisk},{olive,mark=+},{green,mark=o},{violet,mark=|},{brown,mark=diamond}}}
}
\pgfplotsset{resultplotWithSize/.style={%
  resultplot,
  height=\plotheight,
  width=\plotwidth}}
  
 \pgfplotsset{resultKAndF/.style={%
  resultplotWithSize,
  error bars/x dir=none,
  error bars/y dir=both,
  error bars/y explicit,
  xlabel=$k$,
  ylabel=$f$}}

 \pgfplotsset{resultKAndOC/.style={%
  resultplotWithSize,
  error bars/x dir=none,
  error bars/y dir=both,
  error bars/y explicit,
  xlabel=$k$,
  ylabel=Oracle calls}}

 \pgfplotsset{resultTAndF/.style={%
  resultplotWithSize,
  error bars/x dir=none,
  error bars/y dir=both,
  error bars/y explicit,
  xlabel=time step,
  ylabel=$f$,
  every axis plot/.append style={line width=1pt,mark size=1pt},
  no marks}
	}

 \pgfplotsset{resultTAndOC/.style={%
  resultplotWithSize,
  error bars/x dir=none,
  error bars/y dir=both,
  error bars/y explicit,
  xlabel=time step,
  ylabel=Oracle calls,
  every axis plot/.append style={line width=1pt,mark size=1pt},
  no marks}
	}

 \pgfplotsset{resultKAndStddev/.style={%
  resultplotWithSize,
  error bars/x dir=none,
  error bars/y dir=both,
  error bars/y explicit,
  xlabel=k,
  ylabel=Stddev in $\%$}
	}

\pgfplotsset{timeplot/.style={%
  resultplot1,
  error bars/x dir=none,
  error bars/y dir=both,
  error bars/y explicit,
  xlabel=Problem size $n$,
  ylabel=Running time (s)}}

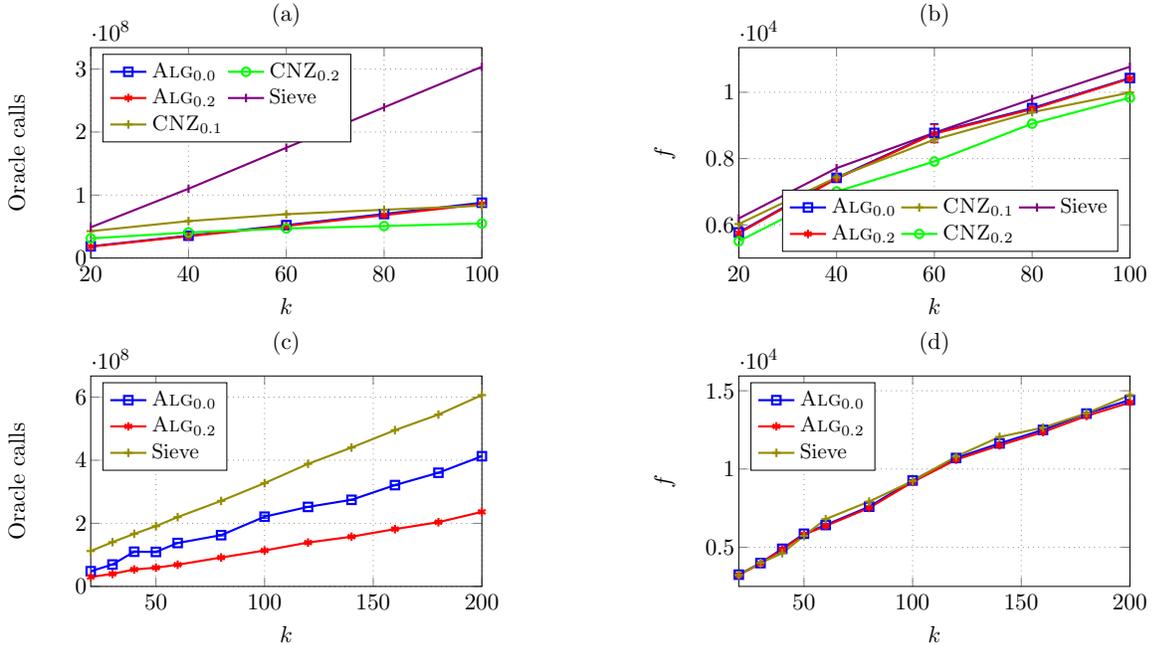
\begin{figure*}[t!]
    \begin{center}
        \begin{subfigure}[b]{0.48 \textwidth}
            \begin{tikzpicture}[scale=0.8]
            \begin{axis}[title=(a), resultKAndOC,name=timeplot1,ymin=0.0,legend pos=north west, legend columns=3]
						\addplot table[x=k,y=OC,y error=stddev] {experiments/results/enr_win_our00_OC.txt};
						\addplot table[x=k,y=OC,y error=stddev] {experiments/results/enr_win_our02_OC.txt};
						\addplot table[x=k,y=OC] {experiments/results/enr_win_cnz01_OC.txt};
						\addplot table[x=k,y=OC] {experiments/results/enr_win_cnz02_OC.txt};
						\addplot table[x=k,y=OC] {experiments/results/enr_win_sieve_OC.txt};
						%\addplot table[x=k,y=OC,y error=stddev] {experiments/results/enr_win_random_OC.txt};
            \legend{$\OurZ$, $\OurTwo$, $\CNZOne$, $\CNZTwo$, Sieve, \Random}
            \end{axis}
            \end{tikzpicture}
        \end{subfigure}
				\hfill
        \begin{subfigure}[b]{0.48 \textwidth}
            \begin{tikzpicture}[scale=0.8]
            \begin{axis}[title=(b), resultKAndF,name=timeplot1,ymin=5000,legend pos=south east, legend columns=2]
						\addplot table[x=k,y=f,y error=stddev] {experiments/results/enr_win_our00_f.txt};
						\addplot table[x=k,y=f,y error=stddev] {experiments/results/enr_win_our02_f.txt};
						\addplot table[x=k,y=f] {experiments/results/enr_win_cnz01_f.txt};
						\addplot table[x=k,y=f] {experiments/results/enr_win_cnz02_f.txt};
						\addplot table[x=k,y=f] {experiments/results/enr_win_sieve_f.txt};
						%\addplot table[x=k,y=f] {experiments/results/enr_win_random_f.txt};
            \legend{$\OurZ$, $\OurTwo$, $\CNZOne$, $\CNZTwo$, Sieve, \Random}
            \end{axis}
            \end{tikzpicture}
        \end{subfigure}
        \begin{subfigure}[b]{0.48 \textwidth}
            \begin{tikzpicture}[scale=0.80]
            \begin{axis}[title=(c), resultKAndOC,name=timeplot1,ymin=0.0,legend pos=north west,
            legend columns = 3]
						\addplot table[x=k,y=OC,y error=stddev] {experiments/results/twi_lts_our00_OC.txt};
            \addplot table[x=k,y=OC,y error=stddev] {experiments/results/twi_lts_our02_OC.txt};
						\addplot table[x=k,y=OC] {experiments/results/twi_lts_sieve_OC.txt};
						%\addplot table[x=k,y=OC,y error=stddev] {experiments/results/twi_lts_random_OC.txt};
            \legend{$\OurZ$, $\OurTwo$, Sieve, \Random}
            \end{axis}
            \end{tikzpicture}
        \end{subfigure}
				\hfill
        \begin{subfigure}[b]{0.48 \textwidth}
            \begin{tikzpicture}[scale=0.8]
            \begin{axis}[title=(d), resultKAndF,name=timeplot1,ymin=2500.0,legend pos=north west,
            legend columns = 3]
						\addplot table[x=k,y=f,y error=stddev] {experiments/results/twi_lts_our00_f.txt};
            \addplot table[x=k,y=f,y error=stddev] {experiments/results/twi_lts_our02_f.txt};
						\addplot table[x=k,y=f] {experiments/results/twi_lts_sieve_f.txt};
						%\addplot table[x=k,y=f,y error=stddev] {experiments/results/twi_lts_random_f.txt};
            \legend{$\OurZ$, $\OurTwo$, Sieve, \Random}
            \end{axis}
            \end{tikzpicture}
        \end{subfigure}
        \caption{\small				
					The plots in this figure are obtained for $f$ being the graph coverage function. Plots (a) and (b) show the results on the Enron dataset. We fix an arbitrary order of the Enron email addresses and process them sequentially over windows of size $30,000$. We first insert all elements, and then delete them in the same order. Plots (c) and (d) depict the results for the ego-Twitter dataset. In this experiment the insertions are performed in a random order, while deletions are performed starting from highest-degree nodes.
					\label{figure:plots-enron}
				}
    \end{center}
\end{figure*}

In this section we empirically evaluate our algorithm.
We perform experiments using a slightly simplified variant of our algorithm; see \cref{appendix:simple} for more details.
We note that this variant also maintains an almost $\nicefrac{1}{2}$-approximate solution after each operation; it differs in the bound on the expected number of oracle queries per update that we can obtain, which is $\tilde{O}(k)$.
A proof of these guarantees can also be found in~\cref{appendix:simple}.

The code of our implementations can be found at \url{https://github.com/google-research/google-research/tree/master/fully_dynamic_submodular_maximization}.
All experiments in this paper are run on commodity hardware.

We focus on the number of oracle calls performed during the computation and on the quality of returned solutions.
More specifically, we perform a sequence of insertions and removals of elements, and after each operation $i$ we output a high-value set $S_i$ of cardinality at most $k$. For a given sequence of $n$ operations, we plot:
%\ifarxiv
	\begin{itemize}
		\item \textbf{Total number of oracle calls} our algorithm performs for each of the $n$ operations.
		\item \textbf{Quality} of the average output set, i.e., $\sum_{i = 1}^n f(S_i) / n$. 
	\end{itemize}
%\else
%	\\
%	-- \textbf{Total number of oracle calls} our algorithm performs for each of the $n$ operations.\\
%	-- \textbf{Quality} of the average output set, i.e., $\sum_{i = 1}^n f(S_i) / n$. 
%\fi

\paragraph{Dominating sets.}
In our evaluation we use the dominating set objective function. Namely, given a graph $G = (V, E)$, for a subset of nodes $Z \subseteq V$ we define $f(Z) = |N(Z) \cup Z|$, where $N(Z)$ is the node-neighborhood of $Z$. This function is monotone and submodular.

%We remark that the underlying graph $G$ is not affected by insertions and removals -- e.g., 
%when a node $v$ is removed, it is no longer legal to use it in a solution, but $v$ remains in $G$ and still yields utility if it is covered by some other node that is legal to use.

\paragraph{Datasets and their processing.} We perform evaluations on the Enron $(|V| =36,692, |E|=183,831)$, the ego-Twitter $(|V| =81,306, |E|= 1,768,149)$, and the Pokec $(|V| =1,632,803, |E|= 30,622,564)$ graph datasets from SNAP Large Networks Data Collection~\cite{leskovec2015snap}. 
% Each of these datasets is a graph. %Enron consists of email accounts and emails exchanged between those accounts. An account is represented as a node, and there is an undirected edge between two accounts $i$ and $j$ if $i$ sent at least one email to $j$.

%We process both datasets as streams of nodes ordered arbitrarily. 
We run two types of experiments on the abovementioned datasets.
\begin{enumerate}
 \item We consider \textbf{a sliding window} of size $\ell$ over an arbitrary order of the nodes of the graph. When the window reaches a node, we add that node to the stream. Similarly, after $\ell$ insertions, i.e., when a node leaves the window, we delete it. This provides us with a stream of interspersed insertions and deletions. Moreover, setting $\ell$ to the number of nodes in the graph is equivalent to inserting all the nodes in an arbitrary order and then deleting them in the same order.  %We report results for various values of $\ell$. \\
\item We insert all the nodes of the graph in \textbf{arbitrary order}. Afterward, we delete them node-by-node by choosing a node in the current solution that has the largest neighborhood. Intuitively, we delete the elements that contribute  the most to the optimum solution; this potentially results in many changes to $\SOL_i$. We observe that even for this stream, our algorithm is efficient and makes a small number of oracle calls on average.
\end{enumerate}

Due to space constraints, we present the results of only two experiments, one for each of the types. For the first type, we present the results on the Enron dataset for a window of size $\ell = 30,000$.
%\footnote{The length of the window is equal to the number $|V|$ of nodes.}
For the second type, we present the results on the ego-Twitter dataset. Further results on other datasets and different values of $\ell$ are included in
\if\fullversion1 \cref{appendix:extra-experiments}. \else the Appendix. \fi

%For a fixed $k$, we also plot how the number of oracle calls and the quality of outputs changes over time.
%\jtodo{move to appendix if moving Figure 2}
%For Enron, at each point we perform submodular maximization on the window of $1,000,000$ email addresses that are the most recently seen on the stream. Observe that this naturally defines insertion and deletion operations; the current element on the stream is inserted, while the one seen $1,000,001$ email addresses ago is deleted. For ego-Twitter, we first insert all the nodes (profiles), and then delete node by node by choosing a node in the current solution that has the largest neighborhood.

\paragraph{The baselines.} We consider the performance of our algorithm for $\epsilon = 0.0$ and $\epsilon = 0.2$, and denote those versions by $\OurZ$ and $\OurTwo$, respectively. Recall that in our algorithm, if an $\epsilon$-fraction of elements is deleted from the solution on some level, we reconstruct the solution beginning from that level. We cannot compare against the true optimum
or the greedy solution,
as computing them is intractable for data of this size. We compare our approach with the following baselines:
\begin{enumerate}
\item The algorithms of \cite{DBLP:journals/corr/ChenNZ16} and \cite{DBLP:conf/www/EpastoLVZ17} (developed concurrently and very similar). This method is designed for the sliding window setting and can only be used if elements are deleted in the same order as they were inserted. It is parametrized by $\epsilon$ and we consider values of $\epsilon = 0.1$ and $\epsilon = 0.2$, and use $\CNZOne$ and $\CNZTwo$ to denote these two variants.
\item $\Sieve$~\cite{badanidiyuru2014streaming}, which is a streaming algorithm that only supports inserting elements. For any insertion, we simply have $\Sieve$ insert the element. For any deletion that deletes an element in the solution of $\Sieve$, we restart $\Sieve$ on the set of currently available elements.\footnote{Like our algorithm, $\Sieve$ operates parallel copies of the algorithm for different guesses of $\OPT$. We restart only those copies whose solution contains the removed element.}
\item \Random algorithm, which maintains a uniformly random subset of $k$ elements. 
\Random outputs solutions of significantly lower quality than other baselines, so due to space constraints we report its objective value results only in the appendix.
\end{enumerate}
%We present results for various values of $\epsilon$ in
%\if\fullversion1
%\cref{appendix:extra-experiments}. 
%\else
%the Appendix.
%\fi
%After any insertion or after a removal which affects the last cover produced by $\Sieve$ we execute $\Sieve$ to obtain a new solution. Recall th  

\paragraph{Results.}
The results of our evaluation are presented in \cref{figure:plots-enron}. As shown in plots (b) and (d), our approach (even for different values of $\epsilon$) is qualitatively almost the same as $\Sieve$. However, compared to $\Sieve$, our approach has a smoother increase in the number of oracle calls with respect to the increase in $k$. As a result, starting from small values of $k$, e.g., $k = 40$, our approach $\OurTwo$ requires at least $2 \times$ fewer oracle calls than $\Sieve$ to output sets of the same quality for both Enron and ego-Twitter.
The behavior of our algorithm for $\epsilon = 0.0$ is closest to $\Sieve$ in the sense that, as soon as a deletion from the current solution occurs, it performs a recomputation (see %\cref{line:S-i-ell-reduced}
Line~7 of \Deletion). For larger $\epsilon$ our approach performs a recomputation only after a number of deletions from the current solution. As a result, for $\epsilon = 0.2$, on some datasets our approach requires almost $3 \times$ fewer oracle calls to obtain a solution of the same quality as $\Sieve$ (see \cref{figure:plots-enron}(c) and (d)). 

Compared to $\CNZOne$ and $\CNZTwo$ in the context of sliding-window experiments (plots (a) and (b) in \cref{figure:plots-enron}), our approach shows very similar performance in both quality and the number of oracle calls. $\CNZTwo$ is somewhat faster than our approach (plot (a)), but it also reports a lower-quality solution (plot (b)).
We point out that $\CNZ$ fundamentally requires that insertions and deletions are performed in the same order. Hence, we could not run $\CNZ$ for plots (c) and (d), where the experiment does not have that special structure.
%The baseline $\Random$ performs significantly fewer oracle calls than any other algorithm that we run, but it also reports solution of significantly lower value. For instance, for $k = 100$ $\Random$ outputs solution of value more than $5$ times smaller than solutions output by other algorithms.
Since our approach is randomized, we repeat each of the experiments $5$ times using fresh randomness; plots show the mean values. The standard deviation of reported values for $\OurZ$ and $\OurTwo$, less than $5\%$,
is plotted in
%For the baseline $\Random$, as expected, standard deviation is larger and reaches around $30\%$ for some values of $k$. The average quality of solution of $\Random$ is several times worse than the one $\Our$, so even taking into account high standard deviation of $30\%$ an output of $\Random$ is significantly outperformed by $\Our$. Standard deviation is plotted in
\if\fullversion1
\cref{fig:steddev-plots}.
\else
the full version.
\fi

%!TEX root = 00-Dynamic Submodular Maximization.tex
%\ifarxiv\else \vspace{-0.8em} \fi
\section{Conclusion and future work}
%\ifarxiv\else \vspace{-0.5em} \fi
We present the first efficient algorithm for cardinality-constrained dynamic submodular maximization, with only poly-logarithmic amortized update time.
We complement our theoretical results with an extensive experimental analysis showing the practical performance of our solution.
Our algorithm achieves an almost $\nicefrac12$-approximation.
A compelling direction for future work is to extend the current result to more general constraints such as matroids.

\ifarxiv\else
	\small
\fi
\bibliography{cite}
\bibliographystyle{alpha}
\if\fullversion1
\ifarxiv\else
	\newpage
\fi
\normalsize 
\onecolumn
%!TEX root = 00-Dynamic Submodular Maximization.tex
\appendix 
%!TEX root = 00-Dynamic Submodular Maximization.tex 
\section{Query complexity of our algorithm} \label{app:complexity}
We now state two invariants  that are maintained by our algorithms. We will then use these invariants to analyze the oracle-query complexity of our algorithms.
\begin{invariant}\label{invariant:size-of-A-ell}
	For any $\ell$, it holds that $|A_{\ell}| \le \numtaus \cdot 2^{T - \ell}$.\footnote{Recall that $T = \log{n}$ (see 
	%\cref{line:definition-of-T}
	Line~3 of $\Init$.)}
	Here we overload notation $A_{\ell}$ to denote $\bigcup_{i=1}^{\numtaus} A_{i,\ell}$.
\end{invariant}

\begin{invariant}\label{invariant:size-of-B-ell}
	For any $\ell$, it holds that $|B_{\ell}| \le 2^{T - \ell}$.
\end{invariant}

\begin{lemma}\label{lemma:level-construct-invariants}
	If \cref{invariant:size-of-A-ell,invariant:size-of-B-ell} hold before invoking \LevelConstruct (\cref{alg:level-construct}), then they also hold after executing %\cref{lvlstop}
Line~13 of \LevelConstruct. Consequently, the invariants hold after \LevelConstruct terminates. 
\end{lemma}
\begin{proof}
	First, notice that the invocation of $\LevelConstruct(\ell')$ never adds new elements to any $B_{\ell}$, hence the claim holds for \cref{invariant:size-of-B-ell}. Now we prove the claim for \cref{invariant:size-of-A-ell}.
	
	Notice that $\LevelConstruct(\ell')$ only potentially increases the elements of level $\ell'$, and the only change for the rest is a reset to the empty set. Therefore we only focus on showing the invariant for  $A_{i, \ell'}$. When $\LevelConstruct(\ell')$ is invoked, it iterates over all $i = 1 \ldots \numtaus$, and for each of them invokes $\BucketConstruct$ on 
	%\cref{line:invoke-bucket-construct}
	Line~8. By the definition, $\BucketConstruct(i', \ell')$ increments $S_{i', \ell'}$ and reduces $A_{i', \ell'}$ until the size of $A_{i', \ell'}$ becomes less than $2^{T - \ell'}$ or until $|\Supto{\numtaus, T}| \ge k$ (see 
	%\cref{line:condition-bucket-construct}
	Line~6 of \BucketConstruct). After this invocation of $\BucketConstruct$ the set $A_{i', \ell'}$ is not changed anymore. There are now two cases, depending on which of the two conditions on 
	%\cref{line:condition-bucket-construct}
	Line~6 of \BucketConstruct is false.
	
	\paragraph{Case $|\Supto{\numtaus, T}| \ge k$.} In this case, each $A_{i, \ell'}$ is set to be the empty set (%
	%\cref{line:reset-all-A-i-ell}
	Line~11 of \LevelConstruct), and hence the claim follows directly.
	
	\paragraph{Case $|A_{i, \ell'}| <  2^{T-\ell'}$.} In this case, the size of $A_{i, \ell'}$ remains at most $2^{T - \ell'}$ throughout the rest of the execution. Since this holds for each $j = 0 \ldots \numtaus$ for which $|A_{j, \ell'}| <  2^{T-\ell'}$, after the loop on
	%\cref{line:level-construct-loop}
	Line~2 terminates we have that $|A_{\ell'}| \le \numtaus 2^{T - \ell'}$.
\end{proof}

\begin{lemma}
	If \cref{invariant:size-of-A-ell,invariant:size-of-B-ell} hold before invoking \Insertion (\cref{alg:insertion}), then they also hold after \Insertion terminates.
\end{lemma}
\begin{proof}
	Observe that
	%\cref{line:insertion-add-e}
	Line~1 of \Insertion changes only sets $B_{\ell}$. Hence, if
	%\cref{line:insertion-large-B-ell}
	Line~3 evaluates to false, then the two invariants still hold. Before we analyze the case when
	%\cref{line:insertion-large-B-ell}
	Line~3 evaluates to true, we first show that there does not exist $\ell'$ such that $|B_{\ell'}| > 2^{T - \ell'}$. Observe that \Insertion is the only function that adds elements to $B_{\ell'}$.
	
	Towards a contradiction, assume that there is an invocation of \Insertion where for some $\ell'$ it holds that $|B_{\ell'}| > 2^{T - \ell'} \ge 1$. Let that be the $c$-th invocation of \Insertion. Since in each invocation of \Insertion the size of $B_{\ell'}$ increases by at most $1$, it means that in the $(c-1)$-st invocation of \Insertion it holds that $|B_{\ell'}| \ge 2^{T  - \ell'}$. Hence,
	%\cref{line:insertion-large-B-ell}
	Line~3 evaluates to true in that invocation. So, by the choice of $\ellstar$ (see
	%\cref{line:choice-of-ellstar}
	Line~4 of \Insertion) it holds that $\ellstar \le \ell'$. But this now implies that
	%\cref{line:insert-reset-all-B-ell}
	Line~6 of $\Insertion$ sets $B_{\ell'}$ to be the empty set. Hence, after
	%\cref{line:insertion-large-B-ell}
	Line~3 of \Insertion evaluates to true in the $(c - 1)$-st invocation, the size of $B_{\ell'}$ in the $c$-th invocation is at most $1$. This contradicts our assumption.
	
	This now implies that when \LevelConstruct is invoked, the two invariants hold. Hence, by \cref{lemma:level-construct-invariants} these two invariants also hold after the execution of \LevelConstruct invoked on
	%\cref{line:insertion-invokes-level-construct}
	Line~7 of \Insertion, and consequently hold after the execution of \Insertion.
\end{proof}

\begin{lemma}
	If \cref{invariant:size-of-A-ell,invariant:size-of-B-ell} hold before invoking \Deletion (\cref{alg:deletion}), then they also hold after \Deletion terminates.
\end{lemma}
\begin{proof}
	On
	%\cref{line:remove-from-A,line:remove-from-B}
	Lines~1 and~2 \Deletion removes some elements from $A_{i, \ell}$ and $B_{\ell}$. So, these steps maintain \cref{invariant:size-of-A-ell,invariant:size-of-B-ell}. The rest of the changes of the sets $A_{i, \ell}$ and $B_{\ell}$ is done through invocation of \LevelConstruct on
	%\cref{line:deletion-level-construct}
	Line~9. The proof now follows by \cref{lemma:level-construct-invariants}.
\end{proof}

\subsection{Oracle-query Complexity of Recomputing a Level}

We begin the analysis by bounding the expected cost of recomputing a level $\ell$ (which entails also recomputing further levels $\ell+1, ...$).

\begin{lemma}[\LevelConstruct Complexity]\label{lemma:running-time-level-construct}
	$\LevelConstruct(\ell)$ performs $O(R^2 \log n \cdot 2^{T-\ell} / \epsilon_p^2)$ oracle queries in expectation, not counting the possible subsequent call to $\LevelConstruct(\ell+1)$.
\end{lemma}
\begin{proof}
    Note that all the oracle queries performed by $\LevelConstruct$ are via invocations of $\BucketConstruct$. Hence, we first analyze the oracle-query complexity of $\BucketConstruct$. We begin by bounding $|A_{i, \ell}|$ during an execution of $\LevelConstruct(\ell)$.
	
	By the invariants, from
	%\cref{line:previous-elements-l1,line:previous-elements-l0}
	Lines~4 and~6 of \LevelConstruct we have that for each $i$ and if $R \ge 2$ it holds that
	\begin{eqnarray}
	|A_{i, \ell}| & \le & |B_{\ell - 1}| + \sum_{j=1}^{R} |A_{j, \ell-1}| \nonumber \\
	& \le & 2^{T - \ell + 1} + R \cdot 2^{T - \ell + 1} \nonumber \\
	& \le & 3 \cdot R \cdot 2^{T - \ell}. \label{eq:bound-on-A-i-ell}
	\end{eqnarray}
        For each $i$, $\LevelConstruct$ on
	Line~8 invokes $\BucketConstruct(i,\ell)$.
	
	\paragraph{Oracle queries of $\BucketConstruct(i, \ell)$.} 
	Line~2 of \BucketConstruct performs at most $|A_{i, \ell}|$ oracle queries.
	Next, \Peeling is (possibly) called, which begins with a call to \EstimateMean.
	For any $a \ge 0$,
	let us denote by $X_a$ the worst-case expected number of oracle queries until the end of the execution of \BucketConstruct, assuming that the algorithm is currently at the beginning of \Peeling
	and
	that the size of the current candidate set is $a$.
	Note that by the current candidate set we mean $N$ if the algorithm is inside \Peeling, and $A_{i,\ell}$ otherwise.
	Our main objective is to bound $X_a$:
	\newcommand{\hep}{\hat{\epsilon_p}}
	\begin{claim} \label{claim:Xa}
	    $X_a \le Q \cdot a$, where $Q := \frac{66 \log(n/\hep)}{\hep^2}$.
	\end{claim}
	\noindent
	We prove \cref{claim:Xa} by induction on $a$.
	To that end, we first make several observations.
	\begin{itemize}
	    \item Call an element of the current candidate set \emph{strong} if its marginal contribution to the current solution is at least $\tau_i$.
	    \item We begin with a call to \EstimateMean. By \cref{lem:estimator}, \EstimateMean makes at most $33 \log(n / \hep) / \hep^2 = Q/2$ oracle queries, and with probability at least $1-\delta$ answers correctly, that is: if the fraction of strong elements in the current candidate set is at least $1-\hep$, then it returns false, and if the fraction is at most $1-2\hat{\epsilon_p}$, then it returns true.
	    \item Exactly one of three events happens:
	    \begin{enumerate}
	        \item fraction of strong elements $\ge 1 - 2 \hep$ and \EstimateMean returns false,
	        \item fraction of strong elements $\le 1 - \hep$ and \EstimateMean returns true,
	        \item fraction of strong elements $> 1 - \hep$ and \EstimateMean returns true, or fraction of strong elements $< 1 - 2 \hep$ and \EstimateMean returns false.
	    \end{enumerate}
	    \item In case of event 1, the algorithm proceeds to sample an element from the current candidate set.
	    With probability at least $1-2\hep$, the sampled element is strong and thus it is picked; thus $a$ is decremented. Otherwise we go back to the state with $a$ candidate elements.
	    \item In case of event 2, the algorithm exits \Peeling and filters $A_{i,\ell}$, removing at least an $\hep$ fraction of its elements. Filtering costs $a$ oracle queries. We continue with at most $\lfloor (1-\hep)a \rfloor$ elements.
	    \item In case of event 3, which has probability at most $\delta = \frac{\hep}{n^2}$,
	    the worst case is that we filter $A_{i,\ell}$ and pick no element, going back to the state with $a$ candidate elements.
	    \item Above we are ignoring the fact that \BucketConstruct/\Peeling may exit if the solution becomes large or the candidate set becomes small. We can ignore this event, as it would be favorable for our analysis (there would be no more oracle queries).
	\end{itemize}
	This gives the following recurrence:
	\begin{alignat*}{2}
	    X_a &\le Q/2 &&+ \prob{\text{event 1}} \cdot \left( (1-2\hep) X_{a-1} + 2 \hep \cdot X_a \right) \\ & &&+ \prob{\text{event 2}} \cdot \left( a + X_{\lfloor (1-\hep)a \rfloor} \right) \\ & &&+ \delta \cdot \left( a + X_a \right)
	    \\
	    &\le Q/2 &&+ \delta a + \delta X_a + \max \left( (1-2\hep) X_{a-1} + 2 \hep \cdot X_a, a + X_{\lfloor (1-\hep)a \rfloor} \right)
	    %\,,
	\end{alignat*}
	which leads to two cases depending on the maximum. In the first case:
	\begin{align*}
	    X_a \le Q/2 + \delta a + \delta X_a + (1-2\hep) X_{a-1} + 2 \hep \cdot X_a %\,,
	\end{align*}
	from which, using the induction hypothesis and rearranging we get
	\begin{align*}
	    (1 - 2 \hep - \delta) X_a &\le Q/2 + \delta a + (1-2\hep) Q (a-1) \\
	    &= Q/2 + \delta a + (1-2\hep) Q a - (1-2\hep) Q \\
	    &= Q/2 + \delta a + \delta Q a + (1-2\hep-\delta) Q a - (1-2\hep) Q \\
	    &= -(1/2 - 2\hep) Q + (Q+1) a \frac{\hep}{n^2} + (1-2\hep-\delta) Q a \\
	    &\le (1-2\hep-\delta) Q a
	\end{align*}
	that is, $X_a \le Q a$ as required.
	(Recall that $\delta$ is defined on Line~1 of \EstimateMean (\cref{alg:estimate_mean}) to be $\hat{\epsilon}_p / n^2$.)
	In the second case, similarly subtracting $\delta X_a$ from both sides,
	\begin{align*}
	    (1-\delta) X_a &\le Q/2 + \delta a + a + X_{\lfloor (1-\hep)a \rfloor} \\
	    &\le Q/2 + \delta a + a + Q \cdot \lfloor (1-\hep)a \rfloor \\
	    &\le Q/2 + \delta a + a + Q \cdot (a - \max(1, \hep a)) \\
	    &\le Q/2 + \delta a + a + \delta Q a - \delta Q a + Q a - Q/2 - Q \hep a / 2 \\
	    &= (1 - \delta) Q a + a \cdot (\delta + 1 + \delta Q - Q \hep / 2) \\
	    &\le (1 - \delta) Q a \,.
	\end{align*}
	This concludes the proof of \cref{claim:Xa}.
	Together with \eqref{eq:bound-on-A-i-ell},
	it implies that the total expected number of oracle queries in one $\BucketConstruct(i,\ell)$ call is at most
	$(Q+1) \cdot 3 \cdot R \cdot 2^{T-\ell} \le O\left( R \cdot 2^{T-\ell} \cdot \log(n/\epsilon_p) / \epsilon_p^2 \right)$.
	Summing this up over all $i=1,...,R$ contributes another $R$ factor. Finally, we can bound $\log(n/\epsilon_p) \le O(\log n)$ using that $\epsilon_p \ge 1/n$.
\end{proof}

\subsection{Frequency of Recomputing a Level}

What remains now is to bound the frequency of level recomputations caused by insertions and deletions.
By a recomputation of level $\ell$ we will refer to Lines 1--12 of $\LevelConstruct$,
that is, everything that happens on level $\ell$ before the recursive call to $\LevelConstruct(\ell+1)$.
Before we proceed to the formal analysis, we explain the main intuitions below.
\label{discussion}
\begin{itemize}
    \item As a recomputation of level $\ell$ costs $\mathrm{polylog}(n) \cdot 2^{T-\ell}$ in expectation,
    showing that it happens at most roughly every $2^{T-\ell}$ operations
    will yield a polylogarithmic amortized complexity per operation.
    \item \textbf{Insertions.} Level $\ell$ can be recomputed due to an insertion if the buffer set $B_\ell$ reaches size $2^{T-\ell}$. At that point $B_\ell$ is emptied, and only grows by at most one element per insertion. Thus it takes $2^{T-\ell}$ insertions to trigger one recomputation.
    \item \textbf{Deletions.} When level $\ell$ is recomputed, robust solution sets $S_{i,\ell}$ are built for all $i$.
    A future deletion from $S_{i,\ell}$ triggers another recomputation
    when $S_{i,\ell}$ has lost at least an $\epsilon$-fraction of its elements.
    However,
    each element in $S_{i,\ell}$ was obtained by uniform sampling from a large set (of size roughly at least $2^{T-\ell}$).
    It is unlikely that an $\epsilon$-fraction of thus-selected elements was deleted
    (by an oblivious adversary)
    unless a similarly large fraction of that large set was deleted.
    Thus, causing a recomputation due to a deletion requires many deletions on average.
    More precisely,
    in \cref{lemma:probably-robust}
    we prove that the probability of a small number of deletions (since the last recomputation) causing $S_{i,\ell}$ to diminish (enough to recompute level $\ell$) is at most $\frac{1}{4R}$.
    \item The above analysis does not yet take into account the interaction between recomputations.
    For example, when $\LevelConstruct(3)$ is invoked,
    level $4$ is recomputed as well.
    This introduces difficulties for the above analysis for $\ell=4$:
    even though the $\LevelConstruct(3)$ call can pay for the current recomputations of all further levels,
    it also causes their solution sets to be resampled.
    Looking from the perspective of $\ell=4$, imagine that such recomputations-from-above were adversarial:
    then an ``adversary'' could cause level $4$ to be resampled at every step until a bad sample is taken -- i.e., one that will make level $4$ be recomputed ``internally'' (due to a deletion on that level) after few steps -- thus making the level-$4$ solution so obtained less robust. And that next recomputation (after few steps) would need to be paid for by level~$4$.
    
    To get around this issue,
    in the proof of \cref{theorem:oracle-complexity} 
    we use an accounting scheme where the level that incurs the initial $\LevelConstruct$ call (whether due to an insertion or a deletion) will pay for the recomputations of all further levels, as well as a constant factor more that goes towards their budgets. In this way, the supposed adversary is weakened, in that trying to artificially force early recomputations on level $4$ as in the example above is now prohibitively expensive.\footnote{Another observation that could help here, which we do not use, is that in our algorithm, further levels clearly depend on earlier levels, but not vice versa.
    Together with the fact that the sequence of operations is fixed in advance (i.e., the adversary is oblivious),
    this means that,
    at the time that a level $\ell$ is recomputed,
    the time at which the next recomputation-from-above will occur is actually already known (not random).
    However, our analysis in \cref{theorem:oracle-complexity} is valid even against an adversary that controls the recomputations-from-above adaptively, knowing the solution that was sampled at level $\ell$ -- as long as causing a recomputation-from-above comes at a cost to the adversary.}
\end{itemize}

We now proceed to the formal statements and proofs.

In the following lemma, which is the crux of our analysis,
we study the deletion-robustness of a bucket $S_{i,\ell}$
in isolation (ignoring any possible recomputations caused by insertions or by other buckets or levels).

\newcommand{\definitionofr}{\ensuremath{r = \frac{\epsilon }{11 R} \cdot 2^{T-\ell}}}
\begin{lemma} \label{lemma:probably-robust}
    Consider a bucket $S_{i,\ell}$ just after being rebuilt with a call to $\BucketConstruct(i,\ell)$.
    Let $e_1, e_2, ..., e_r$ be the sequence of the next $r$ elements deleted beginning from that time,
    where \[ \definitionofr . \]
    (Note that this sequence is not random as we assume an oblivious adversary, i.e., the input is fixed before the algorithm has started.)
    Then \[ \prob{|S_{i,\ell} \cap \{e_1, ..., e_r\}| > \epsilon \cdot |S_{i,\ell}|} \le \frac{1}{4R}, \]
    where the probability is over the randomness used to construct $S_{i,\ell}$.
\end{lemma}
Our proof of \cref{lemma:probably-robust} is a probabilistic coupling argument
in which we relate the process of building $S_{i,\ell}$
to another process that we call the ballot game,
which was first studied in the 19th century~\cite{whitworth,bertrand1887solution}.
It models an election between two candidates where all votes are counted in random order,
and studies the probability that the winner is always ahead (by a given factor) during the count.

\textbf{Ballot game:}
\begin{itemize}
    \item start with $r$ red elements and $2^{T-\ell}(1-2\hat{\epsilon_p}) - r$ non-red elements,
    \item while there are elements remaining, pick one at random,
    \item the ballot game \emph{fails} if at any point the fraction of red elements picked exceeds $\epsilon$.
\end{itemize}

Intuitively, the ballot game is very close to what we want to show in \cref{lemma:probably-robust}, i.e., that the fraction of elements that will be deleted in the next $r$ operations is always below $\epsilon$ during the execution of the algorithm. We begin by listing certain useful properties of the ballot game; then, we formalize the coupling in the proof of \cref{lemma:probably-robust}.

The following follows from known results:
\begin{lemma} \label{lem:ballot-game-fail}
    The probability that the ballot game fails is at most $\frac{1}{5R}$.
\end{lemma}
\begin{proof}
    It is known~\cite{aeppli1924theorie,bartonmallows}
    that in a ballot game with $a_1 = 2^{T-\ell}(1-2\hat{\epsilon_p}) - r$ non-red and $a_2 = r$ red elements,
    the probability that
    the number of counted non-red ballots is always more than
    $\mu = \ceil{1/\epsilon} - 1$ times more
    than the number of  counted red ballots is 
    \[
        \frac{a_1 - \mu a_2}{a_1 + a_2} = \frac{2^{T-\ell}(1-2\hat{\epsilon_p}) - r - (\ceil{1/\epsilon} - 1)r}{2^{T-\ell}(1-2\hat{\epsilon_p})} = 1 - \frac{r \ceil{1/\epsilon}}{2^{T-\ell}(1-2\hat{\epsilon_p})} ,
    \]
    thus
    \begin{align*}
        \prob{\text{(at any time) } \text{red} > \epsilon \cdot \text{total}}
        &=
        \prob{\text{(at any time) } \text{red} \cdot \left(1/\epsilon - 1\right) > \text{non-red}} \\
        &\le
        \prob{\text{(at any time) } \text{red} \cdot \mu \ge \text{non-red}} \\
        &=
        \frac{r \ceil{1/\epsilon}}{2^{T-\ell}(1-2\hat{\epsilon_p})}
        \le
        \frac{r \cdot 2/\epsilon}{2^{T-\ell}(1-2\hat{\epsilon_p})}
        =
        \frac{\epsilon \cdot 2^{T-\ell} \cdot 2/\epsilon}{11 R \cdot 2^{T-\ell}(1-2\hat{\epsilon_p})}
        \le
        \frac{1}{5R} .
    \end{align*}
\end{proof}
We will also need the following intuitive fact:
\begin{fact} \label{fact:red-is-bad}
    Condition on any state of the ballot game. The probability of failure conditioned on picking a red element in the next move is no lower than if a non-red element is picked.
\end{fact}
\begin{proof}[Proof of \cref{lemma:probably-robust}]
    Let us call the elements $e_1, e_2, ..., e_r$ \emph{red}.
    We recall how the bucket $S_{i,\ell}$ is constructed.
    First, it is emptied of old contents (either in Line~8 of $\Deletion$ or in Line~5 of $\Insertion$).
    We start with a set $A_{i,\ell}$, which has size at least $2^{T-\ell}$
    (if not, then $S_{i,\ell}$ remains empty and thus infinitely robust to deletions).
    From this state, we perform some number of \textit{iterations}.
    We say that an iteration concludes when one of the following four events happens:
    \newcommand{\Drop}{\ensuremath{\mathrm{Drop}}}
    \newcommand{\Exit}{\ensuremath{\mathrm{Exit}}}
    \newcommand{\PickNonRed}{\textrm{PickNonRed}}
    \newcommand{\PickRed}{\textrm{PickRed}}
    \begin{enumerate}
        \item \Exit: $\BucketConstruct$ terminates.
        \item \Drop: the set $A_{i,\ell}$ is filtered on Line~2 of $\BucketConstruct$, and it loses at least one element but retains at least $2^{T-\ell}$ many.
        \item \PickRed: a red element is added in $\Peeling$.
        \item \PickNonRed: a non-red element is added in $\Peeling$.
    \end{enumerate}
    Once an iteration concludes (with the earliest of the above four events), a new iteration begins (unless Exit happened).
    
    Because the event Exit could happen at any time, in effect we need to upper-bound the probability that after any iteration,
    the set of currently picked elements has more than an $\epsilon$ fraction of red elements.
    We will say that the algorithm \emph{fails} if this happens.
    
    We upper-bound the probability of failure by coupling the $\BucketConstruct$ execution to an instance of the ballot game. Intuitively, we do so as follows:
    \begin{itemize}
        \item if \Exit: the algorithm cannot fail anymore, and we stop;
        \item if \Drop: the ballot game does nothing;
        \item if \PickRed: the ballot game picks a red element;
        \item if \PickNonRed: the ballot game picks a non-red element.
    \end{itemize}
    We will use \cref{fact:red-is-bad} together with the fact that the algorithm picks a red element with lower probability than the ballot game.
    
    \newcommand{\Aptotal}{A^{\mathrm{picked}}_{\mathrm{total}}}
    \newcommand{\Artotal}{A^{\mathrm{rem}}_{\mathrm{total}}}
    \newcommand{\Apred}{A^{\mathrm{picked}}_{\mathrm{red}}}
    \newcommand{\Arred}{A^{\mathrm{rem}}_{\mathrm{red}}}
    \newcommand{\Bptotal}{B^{\mathrm{picked}}_{\mathrm{total}}}
    \newcommand{\Brtotal}{B^{\mathrm{rem}}_{\mathrm{total}}}
    \newcommand{\Bpred}{B^{\mathrm{picked}}_{\mathrm{red}}}
    \newcommand{\Brred}{B^{\mathrm{rem}}_{\mathrm{red}}}
    Formally, we let the state of the algorithm be described by a tuple of nonnegative integers $(\Aptotal,\Artotal,\Apred,\Arred)$, and similarly a tuple $(\Bptotal,\Brtotal,\Bpred,\Brred)$ describes the state of the ballot game (here $\mathrm{rem}$ stands for "remaining"). We will maintain the following invariants:
    \begin{align}
        \Arred &\le \Brred, \label{inv1} \\
        (1-2\hat{\epsilon_p}) \Artotal &\ge \Brtotal, \label{inv2} \\
        \text{the algorithm has } &\text{not failed yet}, \label{not-fail} \\
        \Aptotal &= \Bptotal, \label{inv3} \\
        \Apred &= \Bpred, \label{inv4} \\
        \Brtotal &= (1-2\hat{\epsilon_p})2^{T-\ell} - \Bptotal, \label{inv3p} \\
        \Brred &= r - \Bpred. \label{inv4p}
    \end{align}
    % one could also add that \Arred \le \Artotal and same for B
    The invariants \eqref{inv3}--\eqref{inv4} are easily seen to be always satisfied, because the algorithm and the ballot game pick elements at the same (both red or both non-red); \eqref{inv3p}--\eqref{inv4p} follow since the ballot game never throws elements away.
    At the beginning, \eqref{inv1} is satisfied because $A_{i,\ell}$ contains at most $r$ red elements (so $\Arred \le r = \Brred$), and \eqref{inv2} follows since $(1-2\hat{\epsilon_p})\Artotal = (1-2\hat{\epsilon_p})|A_{i,\ell}| \ge (1-2\hat{\epsilon_p})2^{T-\ell} = \Brtotal$. We will verify that \eqref{inv1}--\eqref{not-fail} are maintained after every iteration.
    
    \newcommand{\Afail}{A_{\textrm{fail}}}
    \newcommand{\Bfail}{B_{\textrm{fail}}}
    Denote the event of the algorithm failing by $\Afail$
    and
    the event of the ballot game failing by $\Bfail$.
    Roughly, we want to show that $\Afail$ is less likely that $\Bfail$,
    but we need to allow some small slack for the possibility of $\EstimateMean$ returning incorrect answers.
    Formally, our goal will be to prove the following:
    \begin{claim} \label{claim:main}
        Consider a state of the algorithm and a state of the ballot game that satisfy the invariants.
        Then \[ \prob{\Afail} \le \prob{\Bfail} + \delta \cdot \Artotal \,. \]
    \end{claim}
    \noindent
    Since the initial states satisfy the invariants,
    \cref{claim:main} and \cref{lem:ballot-game-fail}
    will imply that
    \[ \prob{\frac{|S_{i,\ell} \cap \{e_1, ..., e_r\}|}{|S_{i,\ell}|} > \epsilon} \le \prob{\Afail} \le \prob{\Bfail} + \delta \cdot n \le \frac{1}{5R} + \frac{\hat{\epsilon_p}}{n} \le \frac{1}{4R} \]
    which concludes the proof of \cref{lemma:probably-robust}; as a reminder, $\delta$ is defined on Line~1 of \EstimateMean (\cref{alg:estimate_mean}) to be $\hat{\epsilon}_p / n^2$.
    
    We will prove the Claim by induction on $\Artotal$.
    We start by taking care of a few special cases:
    \begin{enumerate}
        \item If $\Arred = 0$, then the algorithm cannot fail anymore and the Claim is satisfied.
        \item The former point is implied if $\Artotal = 0$ (induction base case).
        \item If the ballot game has only red elements remaining, then the algorithm also cannot fail anymore, because it has already picked a large number of elements. Namely,
        \[
            (1-2\hat{\epsilon_p})2^{T-\ell} - \Aptotal
            \stackrel{\eqref{inv3}}{=}
            (1-2\hat{\epsilon_p})2^{T-\ell} - \Bptotal
            \stackrel{\eqref{inv3p}}{=}
            \Brtotal
            =
            \Brred
            \stackrel{\eqref{inv4p}}{=}
            r - \Bpred
            \stackrel{\eqref{inv4}}{=}
            r - \Apred
            \,,
        \]
        thus
        \[
            \Aptotal = (1-2\hat{\epsilon_p})2^{T-\ell} - r + \Apred
            \ge (1-2\hat{\epsilon_p})2^{T-\ell} - r
        \]
        and even if the algorithm now picked all remaining red elements, the red fraction would be at most
        \[
            \frac{r}{\Aptotal} \le \frac{r}{(1-2\hat{\epsilon_p})2^{T-\ell} - r} \ll \epsilon \,.
        \]
    \end{enumerate}
    
    \paragraph{Inductive step.} Now we make the inductive step, assuming none of the special cases hold.

    Let us first note that Exit is a very good event, since the algorithm cannot fail anymore:
    \begin{equation}
    \prob{\Afail \mid \Exit} = 0 \,.
        \label{eq:exit-fail2}
    \end{equation}
    
    What about \Drop?
    If \Drop, let $(\Aptotal,\tilde{\Artotal},\Apred,\tilde{\Arred})$ be the new algorithm state after this iteration
    (note that the algorithm did not pick any element in this iteration).
    Then the states $(\Aptotal,\tilde{\Artotal},\Apred,\tilde{\Arred})$ and $(\Bptotal,\Brtotal,\Bpred,\Brred)$
    still satisfy all the invariants.
    We verify \eqref{inv1}: $\tilde{\Arred} \le \Arred \le \Brred$;
    and \eqref{inv2}: $(1-2\hat{\epsilon_p})\tilde{\Artotal} \ge (1-2\hat{\epsilon_p}) 2^{T-\ell} \ge \Brtotal$
    (the last inequality is because $\Brtotal$ never increases).
    Since $\tilde{\Artotal} < \Artotal$,
    the inductive hypothesis implies
    \begin{equation}
       \prob{\Afail \mid \Drop} \le \prob{\Bfail} + \delta \cdot \tilde{\Artotal} \le \prob{\Bfail} + \delta \cdot (\Artotal - 1) \,. \label{eq:drop-fail2}
    \end{equation}
    
    Having \eqref{eq:exit-fail2} and \eqref{eq:drop-fail2},
    let us trace what happens during the iteration.
    First, if the previous iteration ended with picking an element, it could be that now the solution has $k$ elements;
    in that case we immediately Exit and the algorithm cannot fail anymore.
    Otherwise, regardless of how the previous iteration (if any) ended,
    the algorithm will proceed to Line~4 of $\Peeling$,
    where $\EstimateMean$ is called.
    
    Let us say that an element (of the set that the algorithm is picking from, i.e., $N$ in $\Peeling$ or $A_{i,\ell}$ if it is currently outside of $\Peeling$) is \emph{strong} if its marginal contribution to the current solution is at least $\tau_i$.
    We have two main cases depending on the fraction of strong elements.
    
    \textbf{Weak case: fraction of strong elements is less than $1-2\hat{\epsilon_p}$.}
    
    Consider the aforementioned first call to $\EstimateMean$ during this iteration.
    By \cref{lem:estimator}, with probability at least $1-\delta$,
    $\EstimateMean$ returns true.
    If it does, then the algorithm will exit $\Peeling$ and return to Line~2 of $\BucketConstruct$,
    where all non-strong elements (at least one element) will be filtered away.
    Thus we have Drop or Exit (depending on whether there are at least $2^{T-\ell}$ strong elements or not).
    By \eqref{eq:exit-fail2} and \eqref{eq:drop-fail2}, either outcome is good.
    
    To sum up, in the weak case \cref{claim:main} holds since
    \begin{align*}
        \prob{\Afail} &\le \prob{\text{first \EstimateMean returns false}} + \prob{\Afail \mid \text{first \EstimateMean returns true}} \\
        &\le \delta + \prob{\Afail \mid \Drop \text{ or } \Exit} \\
        &\le \prob{\Bfail} + \delta \cdot \Artotal \,.
    \end{align*}
    
    \textbf{Strong case: fraction of strong elements is at least $1-2\hat{\epsilon_p}$.}
    
    As in the previous case, by \eqref{eq:exit-fail2} and \eqref{eq:drop-fail2}, the Drop and Exit outcomes are good.
    Here, we argue that picking an element is also a good outcome.
    Namely, we will show that
    \begin{align}
        \prob{\Afail \mid \PickRed \text{ or } \PickNonRed} \le \prob{\Bfail} + \delta \cdot \Artotal \,. \label{pick-red-or-not-red-fail}
    \end{align}
    Showing \eqref{pick-red-or-not-red-fail} will conclude the proof of \cref{claim:main}.
    
    Note that in \Peeling we always sample an element from the current set
    uniformly at random.
    If it is strong, it is picked (and the iteration ends);
    if not, the algorithm continues.
    Since we are conditioning on $\PickRed \text{ or } \PickNonRed$,
    eventually a strong element will be picked.
    Crucially, this produces \emph{a uniform sample from the set of strong elements}.
    (Note that here we are not using any guarantees about \EstimateMean \ -- it might as well return arbitrary answers.)
    %\footnote{In fact, if \EstimateMean ever returned true, we would have \Drop or \Exit, depending on whether there are at least $2^{T-\ell}$ strong elements or not -- unless \textit{all} elements are strong, in which case we would continue. Though in that case, \EstimateMean can never return true anyway.})
    By this, and since strong elements are at least a $1-2\hat{\epsilon_p}$ fraction, we have
    \begin{equation}
        \prob{\PickRed \mid \PickRed \text{ or } \PickNonRed}
        \le
        \frac{\Arred}{(1-2\hat{\epsilon_p})\Artotal} \,. \label{eq:red-is-small}
    \end{equation}
    We continue by writing
    \begin{align}
        &\prob{\Afail \mid \PickRed \text{ or } \PickNonRed} %\label{eq:total-expectation}
        \nonumber \\ &= 
        \prob{\PickRed \mid \PickRed \text{ or } \PickNonRed} \prob{\Afail \mid \PickRed} \nonumber \\ &+
        \prob{\PickNonRed \mid \PickRed \text{ or } \PickNonRed} \prob{\Afail \mid \PickNonRed} \nonumber
        \,.
    \end{align}
    Consider picking a red element; we want to show
    \begin{equation} \label{eq:pick-red}
        \prob{\Afail \mid \PickRed} \le \prob{\Bfail \mid \text{ballot picks red}}  + \delta \cdot \Artotal \,.
    \end{equation}
    Consider the state of the algorithm and the ballot game if both pick a red element\footnote{The ballot game state is valid since $\Brred \ge \Arred$ by \eqref{inv1}, and $\Arred > 0$ since we are not in special case 1.}.
    We have two cases.
    If the algorithm fails, then the ballot game also fails (by \eqref{inv3},\eqref{inv4} they have the same red fractions),
    so both probabilities in \eqref{eq:pick-red} are 1.
    Otherwise,
    they continue to satisfy all invariants
    (in \eqref{inv1} both sides decrease by 1, and in \eqref{inv2} the right-hand side decreases more than the left-hand-side),
    and we get \eqref{eq:pick-red} by the inductive hypothesis (with a slack of $\delta$).
    
    Consider picking non-red elements\footnote{The ballot game state is valid since we are not in special case 3.} similarly (albeit without the possibility of failure).
    We get
    \begin{equation} \label{eq:pick-non-red}
        \prob{\Afail \mid \PickNonRed} \le \prob{\Bfail \mid \text{ballot picks non-red}}  + \delta \cdot \Artotal \,.
    \end{equation}
    Together we have
    \begin{align}
        &\prob{\Afail \mid \PickRed \text{ or } \PickNonRed}  - \delta \cdot \Artotal \nonumber \\ &\le 
        \prob{\PickRed \mid \PickRed \text{ or } \PickNonRed} \prob{\Bfail \mid \text{ballot picks red}} \nonumber \\ & \ \ \ +
        (1 - \prob{\PickRed \mid \PickRed \text{ or } \PickNonRed}) \prob{\Bfail \mid \text{ballot picks non-red}} \label{eq:conv-comb1}
        \,,
    \end{align}
    which we want to upper-bound by
    \begin{align}
        \prob{\Bfail} &=
        \prob{\text{ballot picks red}} \prob{\Bfail \mid \text{ballot picks red}} \nonumber \\ & \ \ \ +
        (1 - \prob{\text{ballot picks red}}) \prob{\Bfail \mid \text{ballot picks non-red}} \label{eq:conv-comb2}
        .
    \end{align}
    By \cref{fact:red-is-bad} we have \[ \prob{\Bfail \mid \text{ballot picks red}} \ge \prob{\Bfail \mid \text{ballot picks non-red}} \,. \]
    Now, \eqref{pick-red-or-not-red-fail} follows since
    \begin{equation*}
        \prob{\PickRed \mid \PickRed \text{ or } \PickNonRed}
        \stackrel{\eqref{eq:red-is-small}}{\le}
        \frac{\Arred}{(1-2\hat{\epsilon_p})\Artotal}
        \stackrel{\eqref{inv1},\eqref{inv2}}{\le}
        \frac{\Brred}{\Brtotal}
        =
        \prob{\text{ballot picks red}}
    \end{equation*}
    (because we have two convex combinations \eqref{eq:conv-comb1},\eqref{eq:conv-comb2} of terms $\prob{\Bfail \mid \text{ballot picks red}}$ and
    \\
    $\prob{\Bfail \mid \text{ballot picks non-red}}$, and \eqref{eq:conv-comb1} has a smaller coefficient than \eqref{eq:conv-comb2} does in front of the larger term).
    This concludes the proof of \cref{claim:main} and \cref{lemma:probably-robust}.
\end{proof}

Recall that all buckets on a given level are recomputed at the same time (in a \LevelConstruct call).
By union-bounding over them,
we get the following corollary of \cref{lemma:probably-robust}.

\begin{lemma} \label{lemma:probably-robust2}
    Consider a level $\ell$ just after being rebuilt with a call to $\LevelConstruct(\ell)$.
    Let $e_1, e_2, ..., e_r$ be the sequence of the next $r$ elements deleted beginning from that time,
    where \[ \definitionofr . \]
    (Note that this sequence is not random as we assume an oblivious adversary, i.e., the input is fixed before the algorithm has started.)
    Then \[ \prob{(\exists i) \quad |S_{i,\ell} \cap \{e_1, ..., e_r\}| > \epsilon \cdot |S_{i,\ell}|} \le \frac{1}{4} \]
    (where the probability is over the randomness used to construct $(S_{i,\ell})_{i=1,...,R}$).
\end{lemma}

Now we are ready to analyze the amortized complexity of our algorithm, i.e., prove \cref{theorem:oracle-complexity}.
We proceed using the strategy and accounting scheme introduced on page~\pageref{discussion}.

\theoremoraclecomplexity*
\begin{proof}
    \newcommand{\Cost}{X}
    Let $\Cost$ be such that recomputing level $\ell$
    (without counting the subsequent recursive call to $\LevelConstruct(\ell+1)$)
    has expected cost at most $\Cost \cdot 2^{T-\ell}$.
    By \cref{lemma:running-time-level-construct}, we can have $\Cost = O(R^2 \log n / \epsilon_p^2)$.
    
    We now describe our charging scheme.
    When a recomputation event happens,
    the responsible layer will make a payment,
    which is used to cover the recomputation cost as well as pay benefits towards the budgets of other layers (or the same layer).
    We will show that for each layer, the sum of payments made minus the sum of benefits received over the course of the algorithm is small.
    
    The accounting scheme is as follows:
    \begin{itemize}
        \item When \Insertion or \Deletion invokes $\LevelConstruct(\ell)$:
        \begin{itemize}
            \item for $\ell' = \ell, ..., T$: \quad level $\ell$ pays for the cost of recomputing level $\ell'$,
            \item for $\ell' = \ell, ..., T$: \quad level $\ell$ pays $X \cdot 2^{T-\ell'}$ as benefit to level $\ell'$.
        \end{itemize}
        Note that the expected payment (from level $\ell$) is at most $4 \cdot X \cdot 2^{T-\ell}$ by \cref{lemma:running-time-level-construct}.
    \end{itemize}
    Note that if a level $\ell'$ is recomputed 
    due to an insertion or deletion on some level $\ell \le \ell'$,
    then the costs of this recomputation are paid by level $\ell$,
    and additionally level $\ell'$ receives a benefit of $X \cdot 2^{T-\ell'}$.
    
    For every timestep $t$ we have
    \[
        \sum_\ell \textrm{Payment}_{t,\ell} = \sum_\ell \textrm{Cost}_{t,\ell} + \sum_\ell \textrm{ReceivedBenefits}_{t,\ell} \,,
    \]
    thus
    \[
        \sum_{t=1}^n \sum_\ell \textrm{Cost}_{t,\ell} = \sum_\ell \left[ \sum_{t=1}^n \textrm{Payment}_{t,\ell} - \textrm{ReceivedBenefits}_{t,\ell} \right].
    \]
    \cref{theorem:oracle-complexity} will follow if we show that for each level $\ell$,
    \begin{equation} \label{eq:payment-benefits}
        \ee{\sum_{t=1}^n \textrm{Payment}_{t,\ell} - \textrm{ReceivedBenefits}_{t,\ell}} \le O(RXn/\epsilon) = O(R^3 n \log n / (\epsilon \epsilon_p^2))
    \end{equation}
    (as then it only remains to sum over levels $\ell=0,...,T$; recall that $T = \log n$).
    
    Let us fix $\ell$.
    Define $\textrm{InsertionPayment}_{t,\ell}$ (resp.~$\textrm{DeletionPayment}_{t,\ell}$) as the payment from level $\ell$ if there is an insertion (resp.~deletion) at timestep $t$ and level $\ell$.
    We have $\textrm{Payment}_{t,\ell} = \textrm{InsertionPayment}_{t,\ell} + \textrm{DeletionPayment}_{t,\ell}$ (and at most one of the two can be positive).
    
    \textbf{Insertions.}
     Level $\ell$ can be recomputed due to an insertion if the buffer set $B_\ell$ reaches size $2^{T-\ell}$ (see Line~3 of \Insertion). At that point $B_\ell$ is emptied.
    Moreover, at the beginning of the algorithm $B_{\ell}$ was empty, and it is augmented only by $\Insertion$ (by at most one element per insertion).
    This altogether means that it takes at least $2^{T-\ell}$ insertions to trigger one recomputation; then, the expected payment from level $\ell$ is at most $4 \cdot X \cdot 2^{T-\ell}$. Thus we get
    \begin{equation} \label{eq:bound-on-insertions}
        \ee{\sum_{t=1}^n \textrm{InsertionPayment}_{t,\ell}} \le \frac{n}{2^{T-\ell}} \cdot 4 \cdot X \cdot 2^{T-\ell} = 4Xn.
    \end{equation}
    
    \textbf{Deletions.}
    \cref{lemma:probably-robust2} shows that the computed solutions are robust to deletions; however, they can also be recomputed due to deletions on levels $\ell'' < \ell$ or insertions on levels $\ell'' \le \ell$.
    We will show that even if an adaptive adversary controls those,
    we are still fine.
    To that end,
    for a fixed level $\ell$
    let us define $E_t$ as the following quantity, which we want to upper-bound:
    \begin{itemize}
        \item $E_t$ is the expected sum of DeletionPayments minus ReceivedBenefits over the next $t$ timesteps,
        where
        \item we assume that level $\ell$ has just been recomputed (and the expectation is, in particular, over the randomness used in that recomputation),\footnote{$E_t$ does not include the DeletionPayment or ReceivedBenefit for that recomputation.} and
        \item an adversary controls recomputations due to deletions on levels $\ell'' < \ell$ and insertions on levels $\ell'' \le \ell$; the adversary can be adaptive based on the solution $(S_{i,\ell})_{i=1,...,R}$ and plays to maximize $E_t$.
    \end{itemize}
    By construction we have
    \begin{equation}
        \ee{\sum_{t=1}^n \textrm{DeletionPayment}_{t,\ell} - \textrm{ReceivedBenefits}_{t,\ell}} \le E_n \,.
    \end{equation}
    Recall that $\definitionofr$.
    We will show by induction on $t$ that
    \begin{equation}
        E_t \le X 2^{T-\ell} + \frac{5 X 2^{T-\ell}}{r} t .
        \label{eq:bound-on-E}
    \end{equation}
    Once we have this, \eqref{eq:bound-on-insertions}--\eqref{eq:bound-on-E} will imply \eqref{eq:payment-benefits}:
    \begin{align*}
        \ee{\sum_t \textrm{Payment}_{t,\ell} - \textrm{ReceivedBenefits}_{t,\ell}} 
        &\le 4Xn + E_n
        \le 4Xn + X 2^{T-\ell} + \frac{5 X 2^{T-\ell}}{r} n \\
        &\le 4Xn + X n + \frac{55 R X}{\epsilon} n
        \le O(RXn/\epsilon) \,.
    \end{align*}
    We are left with proving \eqref{eq:bound-on-E}.
    First, we state without proof the intuitive fact that $E$ is monotone:
    \begin{equation} \label{eq:E-monotone}
        E_t \le E_{t+1}.
    \end{equation}
    The base case $t=0$ of \eqref{eq:bound-on-E} is trivially satisfied. Fix $t>0$.
    Let $D$ be the minimum number of future deletions that will cause some bucket on level $\ell$
    to lose an $\epsilon$ fraction of its elements;
    by \cref{lemma:probably-robust2} we have $\prob{D \le r} \le \frac 14$.
    We will distinguish two cases.
    \begin{itemize}
        \item
        If $D \le r$, we assume the worst possible outcome, which is that we recompute immediately ($D=1$).
        This incurs a DeletionPayment that is at most $4 \cdot X \cdot 2^{T-\ell}$ in expectation,
        as well as a benefit of $X \cdot 2^{T-\ell}$.
        In total we have at most $3 \cdot X \cdot 2^{T-\ell} + E_{t-1}$.
        \item
        If $D > r$, the adversary could either decide to wait the $D$ timesteps,
        or recompute (and pay us a benefit).
        \begin{itemize}
            \item If the adversary waits: in case $t < r$, nothing else happens and there are no more payments or benefits.
            Otherwise, we incur the DeletionPayment (minus ReceivedBenefit) associated with the next recomputation
            -- at most $3 \cdot X \cdot 2^{T-\ell}$ in expectation --
            and the process restarts after $D$ steps.
            Since $E$ is monotone and $D > r$, we upper-bound the future value by $E_{t-r}$.
            \item  If the adversary is going to recompute, it is best to do so immediately:
            recomputing after having waited some $k < D$ steps yields a value of $- X \cdot 2^{T-\ell} + E_{t-k}$,
            and since $E$ is monotone, $k = 1$ is best.
        \end{itemize}
    \end{itemize}
    By the above, we can write
    \begin{align*}
        E_t
        &\le 
        \prob{D \le r} \left( 3 \cdot X \cdot 2^{T-\ell} + E_{t-1} \right) \\
        &+ \prob{D > r} \max \left( - X \cdot 2^{T-\ell} + E_{t-1}, \left( 3 \cdot X \cdot 2^{T-\ell} + E_{t-r} \right) \cdot \mathbbm{1}_{t \ge r} \right) \\
        &\le
        X \cdot 2^{T-\ell} \cdot \left[
        \prob{D \le r} \left( 3 + 1 + \frac{5}{r} (t-1) \right) \right. \\
        &\qquad \qquad \quad + \left. \prob{D > r} \max \left( - 1 + 1 + \frac{5}{r} (t-1), \left( 3 + 1 + \frac{5}{r} (t-r) \right) \cdot \mathbbm{1}_{t \ge r} \right) \right]
    \end{align*}
    where for the second inequality we applied the induction hypothesis
    (and took the $X \cdot 2^{T-\ell}$ factor out).
    Now we claim that the first argument to the max operator is always larger than the second:
    if $t < r$, then we have $-1+1+ \frac{5}{r} (t-1) \ge 0$,
    and otherwise
    we have $\frac{5}{r}(t-1) \ge 4+\frac{5}{r}(t-r)$ because $\frac{5}{r}(r-1) \ge 4$.
    Thus we continue:
    \begin{align*}
        E_t &\le X \cdot 2^{T-\ell} \cdot \left[
        \prob{D \le r }\left( 4 + \frac{5}{r} (t-1) \right) + \prob{D > r} \frac{5}{r} (t-1) \right]
        \\
        &= X \cdot 2^{T-\ell} \cdot \left[ \frac{5}{r} (t-1) + \prob{D \le r } \cdot 4 \right] \\
        &\le X \cdot 2^{T-\ell} \cdot \left[ \frac{5}{r} (t-1) + 1 \right] \\
        &\le X \cdot 2^{T-\ell} + \frac{5 X 2^{T-\ell}}{r} t \,,
    \end{align*}
    which proves \eqref{eq:bound-on-E} and the theorem.
\end{proof}

%!TEX root = 00-Dynamic Submodular Maximization.tex 
\section{Approximation analysis of our algorithm}\label{app:correctness}
Let us start by introducing a key property of our algorithm. Throughout this section, in case $\ell=0$, by $\Supto{r,\ell-1}$ we denote $\Supto{r-1,T}$ and we define $\Supto{-1,T} = \emptyset$.
\begin{observation}\label{obs:emptysets}
	The only time when elements are added to one of the sets $S_{i,\ell}$ is %\cref{line:bucket-invoke-peeling}
Line~4 of \cref{alg:bucket-construct}. Moreover, all the sets $S_{\cdot,\cdot}$ after $S_{i,\ell}$ are empty, i.e.,  
	\[
	\bigcup_{1 \le j \le \numtaus, \ell < r \leq T} S_{j, r} \cup \bigcup_{i+1 \le j \le \numtaus} S_{j, \ell} = \emptyset.
	\]
\end{observation}
\begin{proof}
	This follows from the fact that we empty all the abovementioned sets in %\cref{line:empty-S-insertion}
Line~5 of \cref{alg:insertion} and %\cref{line:empty-S-deletion}
Line~8 of \cref{alg:deletion} before calling \LevelConstruct, and those are the only lines that \BucketConstruct is called from. 
\end{proof}
\theoremupperbound*
\begin{proof}
Let $t_\ell \le i$ be the largest timestamp until the update $i$ such that $\LevelConstruct(\ell')$, for any $\ell' \le \ell$, was triggered by update $t_\ell$. Intuitively, $t_\ell$ represents the last time when the level $\ell$ was recomputed. We now justify why we consider $\ell' \le \ell$ as opposed to only $\ell' = \ell$. First, observe that $\LevelConstruct$ is invoked only via $\Insertion$ and $\Deletion$. When $\LevelConstruct(\ell')$ is invoked via those methods, then for all $\ell \ge \ell'$ the sets $S_{j, \ell}$ are set to be  empty and their re-computation is potentially triggered via recursive invocations of $\LevelConstruct(\ell')$. This re-computation stops if at some point the current solution contains $k$ elements and hence $\LevelConstruct(\ell)$ might not be invoked. Nevertheless, $S_{j, \ell}$ gets updated, in that case by becoming the empty set. This also implies that $t_{\ell} \le t_{\ell + 1}$ for every $\ell$.

Let $S_{j, \ell}^t$ be the set $S_{j, \ell}$, $A_{j, \ell}^t$ be the set $A_{j, \ell}$ and $B_{\ell}^t$ be set $B_\ell$ at the time after update $t$.
Define
\[
    X \eqdef \bigcup_{0 \le \ell \le T} \bigcup_{1 \le j \le R} S_{j, \ell}^{t_\ell}.
\]

In our analysis, we lower-bound $f(\SOL_i)$ by relating $\SOL_i$ to $X$.
We distinguish two cases depending on the size of $X$:

\begin{itemize}
\item $|X| \ge k$:

We will prove that $|\SOL_i| \ge (1-\epsilon) |X| \ge (1-\epsilon) k$. That directly implies that $f(\SOL_i) \ge (1-\epsilon) \gopt / 2$, as each element added to $\SOL_i$ \emph{deterministically} has marginal gain at least some $\tau$ (this is ensured on Line~8 of \Peeling).
This is with respect to the current solution at the time when the element is added, which is a superset (due to possible deletions between then and time $i$) of the appropriate prefix of $\SOL_i$; by submodularity, its contribution to the appropriate prefix is also at least $\tau$.
By Line~2 of \cref{alg:initialization}, it holds that $\tau \ge \gopt / (2k)$.

To prove that $|\SOL_i| \ge (1-\epsilon) |X|$, observe that $X$ corresponds to the union of the most recent -- relative to the update $i$ -- sets $S_{j, \ell}$ obtained by invoking \LevelConstruct. Hence $\SOL_i$ is obtained from $X$ by removing some elements via \Deletion. Since removing an $\epsilon$-fraction of elements from $S_{j, \ell}$ would have triggered an invocation of $\LevelConstruct(\ell)$ -- see Line~7 of \Deletion -- we have that no more than $\epsilon |S_{j, \ell}|$ elements are removed from $S_{j, \ell}$ since the timestamp $t_{\ell}$. This shows $|\SOL_i| \ge (1-\epsilon) |X| \ge (1-\epsilon) k$.
\footnote{
It might be tempting to rush to a conclusion that $f(X) \ge |X| \gopt / (2k)$, by following the same analogy we used to deduce $f(\SOL_i) \ge |\SOL_i| \gopt / (2k)$. However, $X$ is the union of different $S_{j, \ell}$ sets constructed at \emph{different} times in the algorithms and such conclusion might be flawed. To be more specific, consider two copies $e_1$ and $e_2$ of the same element, i.e., $f(e_1 | e_2) = f(e_2 | e_1) = 0$. Consider the following scenario. The element $e_1$ is included in $S_1^{t_1}$ at time $t_1$. Then, at time $t_1 + 1$ the element $e_1$ is removed from the system, but it does not trigger a re-computation of $S_1$. Assume that at time $t_2 = t_1 + 2$ the element $e_2$ is inserted into the system, and this insertion triggers a re-computation of $S_2$, which concludes by $S_2^{t_2}$ containing $e_2$. Clearly, $f(e_2 | S_1^{t_1}) = 0$, which is exactly the contribution of $e_2$ given the elements added to $X$ previously. On the other hand, $f(e_2 | S_1^{t_2}) \ge \gopt / (2k)$.
}

\item $|X| < k$: 
In this case we will first lower-bound $f(X)$, and then relate $f(X)$ to $f(\SOL_i)$.

We want to show that for any element $e \in O_i$ it holds that $f(e | X) < \frac{\gopt}{2k}$. To that end, consider any element $e \in O_i$.
Let $\ell'$ be the smallest level such that just after processing update $i$ it holds that $e \notin A_{\ell'} \cup B_{\ell'}$; in the sequel we will show that such an $\ell'$ exists.
Let $B_{-1} = V$; this notation is justified by the fact that for $\ell = 0$ the invocation of $\LevelConstruct$ initializes $A_{i, \ell} = V$ (see Line~6). So, letting $B_{-1} = V$ simplifies the notation and ensures $\ell' \ge 0$.

We can also show that $e \notin A_{\ell'}^{t_{\ell'}} \cup B_{\ell'}^{t_{\ell'}}$ and $e \in A_{\ell' - 1}^{t_{\ell'}} \cup B_{\ell' - 1}^{t_{\ell'}}$ in the following way. First, the only way to remove $e$ from some $A_{i, \ell}$ or $B_{\ell}$ is via $\LevelConstruct$ or $\Deletion$. Second, we have that $\Deletion(e)$ was not invoked as $e \in O_i$. Third, $\LevelConstruct$ is not invoked for levels $\ell'$ or $\ell'-1$ after $t_{\ell'}$ by definition. Fourth, to argue that $e \in A_{\ell' - 1}^{t_{\ell'}} \cup B_{\ell' - 1}^{t_{\ell'}}$ and even that $e \in V$ at time $t_{\ell'}$, it suffices to observe that $\Insertion(e)$ was not invoked after $t_{\ell'}$ -- if that was the case, then it would hold that $e \in B_{\ell'}$ as well.

When $\LevelConstruct(\ell')$ was invoked after update $t_{\ell'}$, it holds that $e \in A_{\ell' - 1}^{t_{\ell'}} \cup B_{\ell' - 1}^{t_{\ell'}}$. That further implies that $e$ is considered while updating $A_{\ell'}$ in order to obtain $A_{\ell'}^{t_{\ell'}}$; see Line~4 of \LevelConstruct. The reason why after that invocation of $\LevelConstruct(\ell')$ it holds that $e \notin A_{\ell'}^{t_{\ell'}}$ is because the sets $A_{j, \ell'}$ were pruned via \BucketConstruct invocations on Line~8 of \LevelConstruct and as a result $f\rb{e | \Supto{j, \ell'}^{t_{\ell'}}} \notin [\tau_j, \tau_{j - 1}]$ for all valid $j$.

We now make several simple observations that enable us to conclude that $f\rb{e | \Supto{R, \ell'}^{t_{\ell'}}} < \tau_R = \gamma / (2k)$.
\begin{itemize}
    \item First, $f(e) \le \OPT_i \le \gamma = \tau_0$, i.e., $f(e | Y) \le \tau_0$ for any set $Y$.
    \item Second, due to submodularity we have $f\rb{e | Y} \ge f\rb{e | Y'}$ for any $Y \subseteq Y'$.
In other words, if the marginal gain of $e$ with respect to the so far chosen elements is less than some $\tau_j$, then its marginal cannot be bigger than $\tau_j$ by adding more elements to the current solution.
    \item Third, and the final observation, \BucketConstruct is invoked from \LevelConstruct in the order of decreasing thresholds $\tau$. This property is crucial to ensure that the algorithm does not ``miss'' to test $e$ against the interval it belongs to.
\end{itemize}
By induction on $j=1,...,R$ we can show $f\rb{e | \Supto{j, \ell'}^{t_{\ell'}}} < \tau_j$.
Together these observations lead to $f\rb{e | \Supto{R, \ell'}^{t_{\ell'}}} < \tau_R = \gamma / (2k)$.

We would also like to argue that $f(e | X) \le \gamma / (2k)$. To that end, recall that for every $\ell$ it holds that $t_{\ell} \le t_{\ell + 1}$. Hence, for $\ell \le \ell'$ it holds that $S_{j, \ell}^{t_{\ell'}} \subseteq S_{j, \ell}^{t_{\ell}} \subseteq X$.
Less formally, it means that the last time $S_{\ell}$ got re-computed did not happen after the last time $S_{\ell'}$ got re-computed, and between the timestamps $t_{\ell}$ and $t_{\ell'}$ the set $S_{\ell}$ might have only lost some elements due to removals.
This implies $\Supto{R, \ell'}^{t_{\ell'}} \subseteq X$, which implies $f(e | X) < \gamma / (2k)$.

\paragraph{Existence of $\ell'$.}
It remains to discuss why such $\ell'$ exists.
After update $t_T$, $\LevelConstruct(T)$ must be called, since this is the only way how the condition in %~\cref{lvlstop} 
Line~13 can be false; for recall that we are considering the case when $|X| < k$ (this is the only place where we use this),
and $X$ is a superset of $\Supto{R,T}^{t_T}$ (as argued in the previous paragraph).
This also results in calling $\BucketConstruct(r,T)$ for all $1 \leq r \leq R$ in %\cref{line:level-construct-loop}
Line~2 of \cref{alg:level-construct}. Notice that $\BucketConstruct(r,T)$ stops only if $|A_{r,T}| = 0$ (in %\cref{line:condition-bucket-construct}
Line~6). Since the invocation $\LevelConstruct(T)$ sets $B_T = \emptyset$ as well, this guarantees that such $\ell'$ exists.

\paragraph{Lower-bounding $f(X)$ in terms of $\OPT_i$.}
We just proved that for each $e \in O_i$,
\[
    f(e | X) \leq  \frac{\gopt}{2k}.
\]
Applying  the above inequality for all the elements $e\in O_i$ along with submodularity, we get that
\[
    f(O_i | X) \leq  k\cdot\frac{\gopt}{2k} \leq \frac{\gopt}{2}.
\]
Moreover, by submodularity and monotonicity, we have that
\[
    f(O_i) \leq f(X) + f(O_i | X).
\]
Combining the above two inequalities we get that

\begin{align}\label{cor:secondcase_j}
    f(X) \geq \OPT_i-\frac{\gopt}{2} \geq \OPT_i - \frac{1+\epsilon_p}{2}\OPT_i \geq \frac{1-\epsilon_p}{2}\OPT_i \,.
\end{align}

\paragraph{Comparing $\SOL_i$ and $X$.}
In the rest of this proof, we show that $f(\SOL_i) \geq \rb{1-\epsilon(1+\epsilon_1)}f(X)$. (Recall that we are in the case where $|X| < k$.)

Consider any $S_{r,\ell}$ ($1 \leq  r \leq \numtaus$, $0 \leq \ell \leq T$); we know that when $\LevelConstruct(\ell)$ was called, we had $f\rb{S_{r,\ell}^{t_\ell}|\Supto{r - 1, \ell}^{t_\ell}} \geq \tau_r|S_{r,\ell}^{t_\ell}|$; this again follows because the algorithm only adds elements with marginal gain at least $\tau_r$.\footnote{By $\Supto{0,\ell}$ we denote $\Supto{R,\ell-1}$. Moreover, we let $\Supto{r,-1} = \emptyset$, for any $r$.}
Moreover, at most an $\epsilon$-fraction
of the elements of $S_{r,\ell}^{t_\ell}$ could have been removed; otherwise, Line~7 of \cref{alg:deletion} would trigger.
Note that $S^i_{r,\ell}$ corresponds to $S_{r,\ell}^{t_\ell}$ after these deletions.
The marginal contribution of each element in $S_{r,\ell}^{t_\ell}$ with respect to $\Supto{r-1,\ell}^{t_\ell}$ is at most $\tau_{r-1} = (1+\epsilon_1)\tau_r$; see Line~2 of \cref{alg:bucket-construct}. By submodularity we get that
\begin{align}
f(S^i_{r,\ell}|\Supto{r-1, \ell}^{t_\ell}) &\geq f(S_{r,\ell}^{t_\ell}|\Supto{r-1, \ell}^{t_\ell})-|S_{r,\ell}^{t_\ell}| \cdot \epsilon (1+\epsilon_1) \cdot \tau_r \nonumber \\
& \ge f(S_{r,\ell}^{t_\ell}|\Supto{r-1, \ell}^{t_\ell}) \cdot \rb{1 - \epsilon (1+\epsilon_1)}. \label{eq:S_r_ell'-lower-bound}
\end{align}

Define $\Xupto{r, \ell}$ in a similar way as $\Supto{r, \ell}$
\[
    \Xupto{r, \ell} \eqdef \bigcup_{1 \le j \le \numtaus, 0 \le r < \ell} S_{j, r}^{t_r} \cup \bigcup_{1 \le j \le i} S_{j, \ell}^{t_\ell}.
\]
Now for all $r$ and $\ell$, let $\Supto{r,\ell}^i$ denote the set $\Supto{r,\ell}^{t_\ell}$ after the $i$-th operation.
Considering that $f$ is submodular, and $\Supto{r -1,\ell}^i \subseteq \Supto{r -1,\ell}^{t_\ell} \subseteq \Xupto{r -1,\ell}$,
we get
\begin{eqnarray*}
f(S^i_{r,\ell}|\Supto{r-1, \ell}^i)  & \geq & f(S^i_{r,\ell}|\Supto{r-1, \ell}^{t_\ell}) \\
 & \stackrel{\eqref{eq:S_r_ell'-lower-bound}}{\geq} & f(S_{r,\ell}^{t_\ell}|\Supto{r - 1,\ell}^{t_\ell}) \rb{1-\epsilon(1+\epsilon_1)} \\
 & \geq & f(S_{r,\ell}^{t_\ell}|\Xupto{r - 1,\ell}) \rb{1-\epsilon(1+\epsilon_1)}.
\end{eqnarray*}
By adding up the above marginal values over $r$ and $\ell$, we get
\begin{align}
f(\Supto{\numtaus,T}^i) \geq \rb{1-\epsilon(1+\epsilon_1)} f(X). 
\end{align}

So $f(\SOL_i) \geq \rb{1-\epsilon(1+\epsilon_1)} f(X).$ This along with \cref{cor:secondcase_j} concludes the proof:
\begin{align*}
f(\SOL_i) &\geq (1-\epsilon_p-\epsilon(1+\epsilon_1)) \frac{\OPT_i}{2}.
\end{align*}
\end{itemize}
\end{proof}

%!TEX root = 00-Dynamic Submodular Maximization.tex
\section{\EstimateMean subroutine}\label{app:peeling}

The \EstimateMean subroutine~\cite{FMZSODA} is a standard unbiased
estimator for the mean of a Bernoulli distribution.
Since $\cD_t$ is a uniform distribution over indicator random
variables, it is in fact a Bernoulli distribution.
The guarantees of in \Cref{lem:estimator} are consequences of Chernoff bounds.

\begin{algorithm}
  \caption{\EstimateMean}
  \label{alg:estimate_mean}
  \vspace{0.1cm}
  \textbf{Input:} access to a Bernoulli distribution $\mathcal{D}$,
    error $\hat{\epsilon_p}$
  \begin{algorithmic}[1]
    \STATE Set failure probability
        $\delta \leftarrow \hat{\epsilon_p}/n^2$
    \STATE Set number of samples
      $m \leftarrow 16 \ceil{\log(2/\delta)/\hat{\epsilon_p}^2}$
    \STATE Sample $X_1, X_2,\dots, X_m \sim \mathcal{D}$
    \STATE Set $\overline{\mu} \leftarrow \frac{1}{m} \sum_{i=1}^m X_i$
    \IF{$\overline{\mu} \le 1 - 1.5\hat{\epsilon_p}$}
      \STATE \textbf{return} \texttt{true}
    \ENDIF
    \STATE \textbf{return} \texttt{false}
  \end{algorithmic}
\end{algorithm}

\begin{lemma}\label{lem:estimator}
For any Bernoulli distribution $\mathcal{D}$,
\EstimateMean uses $O(\log(\hat{\epsilon_p}^{-1}n)/\hat{\epsilon_p}^{2})$
samples to correctly report one of the following properties
with probability at least
$1 - \delta$:
\begin{enumerate}
  \item If the output is \textnormal{\texttt{true}}, then the mean of $\mathcal{D}$ is $\mu \le 1 - \hat{\epsilon_p}$.
  \item If the output is \textnormal{\texttt{false}}, then the mean of $\mathcal{D}$ is $\mu \ge 1 - 2\hat{\epsilon_p}$.
\end{enumerate}
Here $\delta$ is as defined in Line~1 of algorithm \EstimateMean. 
\end{lemma}
\begin{proof}
By construction, the number of samples is
$m = 16 \ceil{\log(2/\delta)/\hat{\epsilon_p}^2} \le 33 \log(n / \hat{\epsilon_p}) / \hat{\epsilon_p}^2$.
To show the correctness of \EstimateMean,
it suffices to prove that
$\Prob{}{\abs*{\overline{\mu} - \mu} \ge \hat{\epsilon_p}/2} \le \delta$.
Letting $X = \sum_{i=1}^m X_i$, this is equivalent to
\[
  \Prob{}{\abs*{X - m\mu} \ge \frac{\hat{\epsilon_p} m}{2}}
  \le \delta.
\]

Using the Chernoff bounds in \Cref{lem:chernoff} and a union bound,
for any $a > 0$ we have
\[
  \Prob{}{\abs*{X - m\mu} \ge a} \le
  e^{-\frac{a^2}{2m\mu}} + e^{-a \min\parens*{\frac{1}{5}, \frac{a}{4m\mu}}}.
\]
Let $a = \hat{\epsilon_p} m /2$ and consider the exponents of the two terms 
separately.
Since $\mu \le 1$, we bound the left term by
\[
  \frac{a^2}{2m\mu} = \frac{\hat{\epsilon_p}^2 m^2}{8m\mu}
  \ge \frac{\hat{\epsilon_p}^2}{8\mu} \cdot \frac{16 \log(2/\delta)}{\hat{\epsilon_p}^2}
  \ge \log(2/\delta).
\]
For the second term, first consider the case when $1/5 \le a/(4m\mu)$.
For any $\hat{\epsilon_p} \le 1$, it follows that
\[
  a \min\parens*{\frac{1}{5},\frac{a}{4m\mu}} = \frac{1}{5}
  \ge \frac{\hat{\epsilon_p}}{10} \cdot \frac{16\log(2/\delta)}{\hat{\epsilon_p}^2}
  \ge \log(2/\delta).
\]
Otherwise, we have $a/(4m\mu) \le 1/5$, and
by the previous analysis we have
$a^2/(4m\mu) \ge \log(2\delta)$.
Therefore, in all cases we have
\[
  \Prob{}{\abs*{X - m\mu} \ge \frac{\hat{\epsilon_p} m}{2}} 
  \le 2e^{-\log(2/\delta)}\\
  = \delta,
\]
which completes the proof.
\end{proof}

\begin{lemma}[Chernoff bounds, \cite{bansal2006santa}]
\label{lem:chernoff}
Suppose $X_1,\dots,X_n$ are binary random variables such that
  $\Prob{}{X_{i}=1} = p_i$. Let $\mu = \sum_{i=1}^n p_i$ and
$X = \sum_{i=1}^n X_i$. Then for any $a > 0$, we have
\[
  \Prob{}{X - \mu \ge a} \le e^{-a \min\parens*{\frac{1}{5}, \frac{a}{4\mu}}}.
\]
Moreover, for any $a > 0$, we have
\[
  \Prob{}{X - \mu \le - a} \le e^{-\frac{a^2}{2\mu}}.
\]

\end{lemma}

%!TEX root = 00-Dynamic Submodular Maximization.tex
\section{Combining the ingredients}
\label{sec:combining}
We now recall well-known techniques that can be used to remove the assumptions we made while designing our algorithm. First, we assumed that we have a tight estimate of $\OPT$, e.g., the value of $\gopt$ in \cref{theorem:upper-bound}. This assumption can be removed by considering geometrically increasing guesses $\gopt = (1+\epsilon_p)^i$ of $\OPT$, and for each of the guesses executing a separate instance of our algorithm. Even though $\OPT$ potentially changes from operation to operation, at each point one of the guesses is correct up to a small multiplicative factor. A similar approach was employed in several prior works. After every operation, we return the maximum-value solution over all $\gopt$'s. \cref{theorem:upper-bound} shows that, for the $\gopt$ value that is close to the true optimum value at that time, the instance parametrized by $\gopt$ returns a solution of high value. Moreover, the number of oracle calls is independent of the value of $\gopt$. This results in losing a factor $\log_{1+\epsilon_p}(k\Delta/(\delta\epsilon_p))$ in the number of oracle calls, where $\Delta, \delta$ denote the value of the elements of maximum and minimum value in the universe, respectively. We do not need to know these two values in advance; we simply compute them on the fly and run parallel copies of the algorithm for the currently relevant guesses of $\OPT$. Moreover, we can again use a simple technique to remove the dependency on $\Delta/\delta$. Namely, it suffices to add an element $e$ to those copies of the algorithm where $f(e) \leq \gopt \leq \frac{k}{\epsilon_p}f(e)$ since: (i) while $e$ is not deleted it holds $\OPT \geq f(e)$, therefore we do not need to consider copies with $\gopt < f(e)$; and, (ii) all elements with $\gopt \geq \frac{k}{\epsilon_p}f(e)$ contribute at most only $\epsilon_p\gopt$ to the solution of this copy. Therefore, this increases the number of oracle calls by a factor of $\log_{1+\epsilon_p}(k/\epsilon_p) \le O(\log(k)/\epsilon_p)$,
while decreasing the approximation guarantee by a factor $1-\epsilon_p$.

Second, we assumed that we know the length $n$ of the stream, which is used to upper-bound $|V_t|$. We remove this assumption as follows. We maintain an upper-bound $\tn$ on $n$. The algorithm is initiated by $\tn = 1$. If at some point $n$ equals $\tn$, we restart the algorithm by doubling $\tn$, i.e, by letting $\tn \leftarrow 2 \cdot \tn$, and defining $T = \log{\tn}$ in \cref{alg:initialization}. This affects the number of oracle calls only by a constant factor, and has no effect on the approximation guarantee.

By \cref{theorem:oracle-complexity}, the amortized expected number of oracle queries per update for the base algorithm was $O\left(\frac{\numtaus^3 \log^2 n}{\epsilon_p^2 \cdot \epsilon}\right)$, where $\numtaus = \log_{1+\epsilon_1} (2k) \le O(\log(k)/\epsilon_1)$,
so now it becomes $O\left(\frac{\log^3 k \log^2 n}{\epsilon_1^3 \cdot \epsilon_p^2 \cdot \epsilon} \cdot \frac{\log k}{\epsilon_p} \right)$.
This concludes the proof of our main result:
\mainresult*

%!TEX root = 00-Dynamic Submodular Maximization.tex
\section{Additional experiments}
\label{appendix:extra-experiments}

\input{151-appendix-plots.tex}

\subsection{Value of $f$ after each operation.} \cref{figure:plots-enron} shows the average value of $f$ for different values of $k$ over all operations (insertions/deletions). In \cref{fig:enron-w=30000-k=20-each-operation} we present a more detailed view of the experiment of \cref{figure:plots-enron}~(b) in the following sense. We fix $k = 20$ and split the operations into $400$ equal-sized blocks. For each block, we compute the average value of $f$. We plot those values for all the baselines. This experiment allows us to compare our approach and the baselines over the course of the entire execution. We can see that $\Our$ is very similar to $\CNZOne$ in each block, while $\CNZTwo$ shows around $10\%$ worse performance in the blocks where $f$ has the highest value. In those blocks, $\Sieve$ has around $5\%$ better performance than $\Our$ and $\CNZOne$.
%We refer the reader to \if\fullversion1 \cref{appendix:extra-experiments} \else the Appendix \fi for additional experiments of this kind.
%\jtodo{move to Appendix?}

\begin{figure}[t!]
    \begin{center}
            \begin{tikzpicture}[scale=0.8]
            \begin{axis}[resultTAndF,name=timeplot1,ymin=0.0,legend pos=outer north east, legend columns=3]
						\addplot table[x=t,y=f] {experiments/results/time_enr_win_our00_f.txt};
						\addplot table[x=t,y=f] {experiments/results/time_enr_win_our02_f.txt};
						\addplot table[x=t,y=f] {experiments/results/time_enr_win_cnz01_f.txt};
						\addplot table[x=t,y=f] {experiments/results/time_enr_win_cnz02_f.txt};
						\addplot table[x=t,y=f] {experiments/results/time_enr_win_sieve_f.txt};
						%\addplot table[x=t,y=f] {experiments/results/time_enr_win_random_f.txt};
            \legend{$\OurZ$, $\OurTwo$, $\CNZOne$, $\CNZTwo$, Sieve, \Random}
            \end{axis}
            \end{tikzpicture}
        \caption{
				This plot presents a more detailed analysis of \cref{figure:plots-enron}~(b) for $k=20$; recall that this corresponds to a window-experiment results performed on Enron. The entire execution (i.e., the performed operations) used to obtain the point in \cref{figure:plots-enron}~(b) for $k=20$ is divided into $400$ blocks. This plot shows the average value of $f$ within each block.
				%\jtodo{move to appendix?}
				\label{fig:enron-w=30000-k=20-each-operation}
        }
    \end{center}
\end{figure}
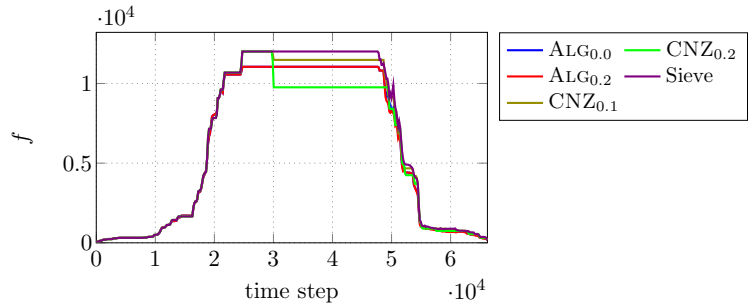

%!TEX root = 00-Dynamic Submodular Maximization.tex
\section{Details of the implemented algorithm}
\label{appendix:simple}
In this section we explain a simpler version of our algorithm, which we use for the implementation and the experiments. We do this as we believe that the bucketing idea is not crucial in real-world applications, even though it is needed to achieve the theoretical guarantee. Namely, given the random structure of layers described in our algorithm, it is very hard for an adversary to delete good elements in a layer without deleting many elements, as the notion of “good” (i.e., contribution of the element with respect to previously chosen elements) heavily depends on the random elements chosen in the previous layers. We believe this situation does not happen often in practice, and so we can simply treat all the elements in a layer in the same way. The only difference is that in this algorithm we do not partition the elements with respect to their contribution and only consider one bucket in each layer. In what follows, we present the algorithm in details for the sake of completeness. One can also read this section without reading the main algorithm.

Similar to before, in this section we assume we know $\OPT$ and first present an algorithm for the case when deletions appear only after all the insertions have been performed. Later, we explain a slight modification to our algorithm to support fully dynamic insertions and deletions. In the following we refer to the set of inserted elements as $V$.

We will construct a hierarchy of sets $H_1, H_2, ...$.
Let $H_1$ be all the elements $e\in V$ such that $f(e) \geq \frac{\OPT}{2k}$, where $\OPT$ is the value of the optimum solution of the set $V$. We first define
\[
	H_1 \leftarrow \setc{e\in V}{f(e) \geq \frac{\OPT}{2k}} \,.
\] 
Afterwards, we select an element $e_1$ from $H_1$ uniformly at random and we define $H_2$ as
\[
H_2 \leftarrow \setc{e\in H_{1}}{f(e|\{e_1\}) \geq \frac{\OPT}{2k}} \,.
\] 
 We repeat this procedure until either we have chosen $k$ elements or for some $1\leq \ell \leq k$, the set $H_{\ell}$ has become empty. More precisely, for any $2 \leq \ell \leq k$, we define 
\[
H_\ell \leftarrow \setc{e\in H_{\ell-1} }{ f(e|\{e_1,\dots,e_{\ell-1}\}) \geq \frac{\OPT}{2k} } \,.
\] 
Then we let $e_\ell$ be a random element from $H_\ell$. We call this simple procedure \emph{a round of peeling}. It is easy to see that the running time of this simple procedure is $O(k|V|)$. This means that the average running time is $O(k)$ per element. 
Let $m$ be the number of elements we have picked this way.
Now we can observe that the solution $E= \{e_1,\dots,e_{m}\}$ is a $1/2$-approximation:
\begin{itemize}
	\item If $|E| = k$, then simply by adding the marginal contributions we get that:
	\[
		f(E) = \sum_{1 \leq \ell \leq k} f(e_\ell|e_1\dots e_{\ell-1}) \geq k\frac{\OPT}{2k} \geq \frac{\OPT}{2}.
	\]
	\item If $|E| < k$, the contribution of any element $e \in V$ to a subset of $E$ is at most $\frac{\OPT}{2k}$ which by submodularity of function $f$ shows that $f(e|E) \leq \frac{\OPT}{2k}$. By summing up this inequality for all the elements $e$ in some optimal solution $O$, from submodularity we get that
	\[
		f(O|E) \leq \frac{\OPT}{2}.
	\] 	
	Now, by monotonicity, it follows that $\OPT \leq f(O|E) + f(E)$ and hence
	\[
		f(E) \geq \frac{\OPT}{2}.
	\]
\end{itemize}
Now consider the deletion of an element $e$. If $e \notin E$, then we do nothing since the current solution has the desired guarantees; we simply remove $e$ from all the sets in our hierarchy. Assume instead that an element $e_\ell \in E$ is deleted. Then we need to recompute the hierarchy beginning from $H_\ell$. Fortunately, this does not change the expected oracle-query complexity much as, intuitively, one needs to delete roughly $O(|H_\ell|)$ elements to delete a fixed element $e_\ell$ with large probability (since it is a randomly chosen element of $H_{\ell}$). Moreover, the time needed for this re-computation is $O(k \cdot |H_\ell|)$, which results in an amortized running time of $O(k)$.
% (In \if\fullversion1 \cref{app:complexity}, \else the Appendix, \fi we make this argument formal for our main algorithm.) This is the main intuition behind our data structure. In the next section we introduce certain additional techniques to improve the runtime to $\poly(\log k)$, which is an exponential improvement over $O(k)$. 
We introduce some additional optimizations.

\paragraph{Improving the insertion algorithm} We introduce two modifications to the previously described algorithm. 
\begin{itemize}
\item {\bf Capping the number of elements in $H$.} We ensure that the number of elements in $H_\ell$ is at most $2^{\log|V| - \ell}$, which also guarantees that the height of the hierarchy is at most $\log |V|+1$. We achieve this by running more than one round of peeling. More precisely, when constructing $H_\ell$, we keep running rounds of peeling  until either its size becomes below the desired threshold (i.e., $|H_\ell| \leq 2^{\log|V| - \ell}$), or we have selected $k$ elements (i.e., $|E| = k$). This does not hurt the running time nor the approximation guarantee.  

\item We maintain buffer sets $B_1,...,B_T$, where $T = \log n +1$ is the maximum possible height of the hierarchy. When an element is inserted, we add it to all the sets $B_\ell$. When, for any $\ell$, the sizes of $B_\ell$ and $H_\ell$ are equal we add the elements of $B_\ell$ to $H_\ell$ (and empty the set $B_\ell$). The main goal of this procedure is to handle updates lazily. During the execution of the algorithm, $B_\ell$ essentially represents those elements that we have not considered in the construction of $H_\ell$. Note that, as described before, the running time for constructing $H_\ell$ is at most $O(k |H_\ell|)$. This guarantees that the amortized running time for insertion is also $O(k)$. Also observe that we merge $B_T$ whenever its size is one. This enables us to maintain a good approximate solution at all times.\footnote{Here we do not provide a formal proof, since in the previous sections we present a more efficient algorithm that obtains better running times for both insertions and deletions.}
\end{itemize}

\paragraph{Improving the deletion algorithm}
Deletions are also handled in a lazy manner: we update our solution only when an $\epsilon$-fraction of elements in a set is deleted. In the next section we explain this idea in more details. Intuitively, we maintain a partitioned version of $H_\ell$ into sets $A_{i,\ell}$ consisting of elements with similar marginal contributions. When an $\epsilon$-fraction of elements in one set $A_{i,\ell}$ is deleted, we trigger a recomputation.  Interestingly, same at the proofs presented before we can show that this significantly reduces the number of re-computations while giving a slightly weaker approximation guarantee, i.e., almost $1/2-\epsilon$. %Finally, we slightly modify the hierarchy to ensure that the sizes of the sets are geometrically decreasing by a factor two. %We intuitively, by carefully picking some elements (potentially more than one) randomly insure that the height of hierarchy is logarithmic.
%Moreover, we also we divide the elements of $H_i$ into sets $A_{i,j}$ of similar contributions. Then when $\epsilon$ fraction of elements in one of $A_{i,j}$ is deleted we recompute. This significantly reduces the number of re-computations while giving a slightly weaker approximation guarantee, i.e., almost $1/2-\epsilon$. We explain this idea in more details in \Peeling Algorithm.

\fi

\end{document}